\documentclass[11pt, a4paper]{article}
\usepackage{amsmath}
\usepackage{graphicx}
\usepackage{enumerate}
\usepackage{natbib}
\usepackage{url} 

\newcommand{\blind}{1}

\usepackage[margin=1in]{geometry}

\usepackage{microtype}
\usepackage{graphicx}
\usepackage{subfigure}
\usepackage{booktabs} 
\RequirePackage{algorithm}
\RequirePackage{algorithmic}
\usepackage{hyperref}

\usepackage{amsmath}
\usepackage{amssymb}
\usepackage{mathtools}
\usepackage{amsthm}

\usepackage[capitalize,noabbrev]{cleveref}

\theoremstyle{plain}
\newtheorem{theorem}{Theorem}[section]
\newtheorem{proposition}[theorem]{Proposition}
\newtheorem{lemma}[theorem]{Lemma}
\newtheorem{corollary}[theorem]{Corollary}
\theoremstyle{definition}
\newtheorem{definition}[theorem]{Definition}

\theoremstyle{remark}
\newtheorem{remark}[theorem]{Remark}

\usepackage[textsize=tiny]{todonotes}

\usepackage[utf8]{inputenc} 
\usepackage[T1]{fontenc}    
\usepackage{hyperref}       
\usepackage{url}            
\usepackage{booktabs}       
\usepackage{amsfonts}       
\usepackage{nicefrac}       
\usepackage{microtype}      
\usepackage{xcolor}         

\usepackage{url}
\usepackage{bm}
\usepackage{multicol}
\usepackage{hhline}

\usepackage{cleveref}
\usepackage{setspace}
\usepackage{ulem}
\usepackage{booktabs}
\usepackage{color}
\usepackage{xcolor}
\usepackage{enumerate}
\usepackage{ulem}
\usepackage{verbatim}
\usepackage{latexsym}
\usepackage{float}
\numberwithin{equation}{section}
\newtheorem*{theorem*}{Theorem}
\newtheorem{fact}{Fact}
\usepackage[most]{tcolorbox}
\usepackage{lipsum}

\usepackage{bbm}

\newcommand{\w}{{\mathbf{w}}}
\newcommand{\x}{{\mathbf{x}}}

\newcommand{\E}{{\mathbb{E}}}
\renewcommand{\P}{{\mathbb{P}}}

\newcommand{\com}[1]{\textcolor{blue}{\texttt{[#1]}}}

\newcommand{\Var}{{\rm Var}}

\newcommand{\tX}{{\widetilde{X}}}

\renewcommand{\epsilon}{\varepsilon}

\renewcommand{\P}{{\mathbb{P}}}

\renewcommand{\P}{\mathbb{P}}

\usepackage{tikz-cd}
\usepackage{tikz}
\usetikzlibrary{matrix,arrows,decorations.pathmorphing}
\usepackage{mathtools}
\usepackage{hyphenat}
\usepackage{multirow}

\def\getangle(#1)(#2)#3{%
  \begingroup%
    \pgftransformreset%
    \pgfmathanglebetweenpoints{\pgfpointanchor{#1}{center}}{\pgfpointanchor{#2}{center}}%
    \expandafter\xdef\csname angle#3\endcsname{\pgfmathresult}%
  \endgroup%
}

\begin{document}

\def\spacingset#1{\renewcommand{\baselinestretch}%
{#1}\small\normalsize} \spacingset{1}


\if1\blind
{
  \title{\bf Edgeworth Accountant: An Analytical Approach\\ to Differential Privacy Composition\thanks{Accepted at JASA}
  }
  \author{Hua Wang$^{1}$ \and Sheng Gao$^{1}$ \and Huanyu Zhang$^{2}$ \and Milan Shen$^{2}$ \and Weijie Su$^{1}$ \and Jiayuan Wu$^{1}$}
\date{
{
$^1$Department of Statistics and Data Science, University of Pennsylvania\\
$^2$Meta Platforms, Inc.\\
\vspace{4ex}
}}
  \maketitle
} \fi




\if0\blind
{
  \bigskip
  \bigskip
  \bigskip
  \begin{center}
    {\LARGE\bf Edgeworth Accountant: An Analytical Approach\\ to Differential Privacy Composition}
\end{center}
  \medskip
} \fi

\bigskip
\begin{abstract}

In privacy-preserving data analysis, many procedures and algorithms are structured as compositions of multiple private building blocks. As such, an important question is how to efficiently compute the overall privacy loss under composition. This paper introduces the Edgeworth Accountant, an analytical approach to composing differential privacy guarantees for private algorithms. Leveraging the $f$-differential privacy framework~\citep{dong2019gaussian}, the Edgeworth Accountant accurately tracks privacy loss under composition, enabling a closed-form expression of privacy guarantees through privacy-loss log-likelihood ratios (PLLRs). As implied by its name, this method applies the Edgeworth expansion to estimate and define the probability distribution of the sum of the PLLRs. Furthermore, by using a technique that simplifies complex distributions into simpler ones, we demonstrate the Edgeworth Accountant's applicability to any noise-addition mechanism. Its main advantage is providing $(\epsilon, \delta)$-differential privacy bounds that are non-asymptotic and do not significantly increase computational cost. This feature sets it apart from previous approaches, in which the running time increases with the number of mechanisms under composition. We conclude by showing how our Edgeworth Accountant offers accurate estimates and tight upper and lower bounds on $(\epsilon, \delta)$-differential privacy guarantees, especially tailored for training private models in deep learning and federated analytics.


\end{abstract}


\noindent%
{\it Keywords:}  Differential Privacy, $f$-Differential Privacy, Edgeworth Expansion, Privacy Accounting
\vfill

\newpage
\spacingset{1.9} 

\section{Introduction}
\label{sec:introduction}

The framework of differential privacy (DP) serves as a mathematically rigorous tool for analyzing and developing algorithms that maintain the privacy of datasets containing sensitive individual information~\citep{dwork2006calibrating}. However, this framework encounters challenges in analyzing the privacy loss of complex algorithms, such as those used in privacy-preserving deep learning and federated analytics~\citep{googleblog2020,wang2021federated}. These complex algorithms are typically composed of simpler private building blocks. A central question in this field is understanding how the overall privacy guarantees deteriorate due to the repeated application of these simple algorithms on the same dataset.

To address this question, there have been developments in relaxing DP and advancing privacy analysis techniques~\citep{dwork2010boosting, dwork2016concentrated, bun2018composable, bun2016concentrated}. This area of research has garnered significant attention following the work of \cite{abadi2016deep}, which introduced the moments accountant technique. This technique provides upper bounds on the overall privacy loss incurred during the training of private deep learning models over multiple iterations. However, the privacy bounds provided by moments accountant are generally not tight, despite being computationally efficient. This limitation arises because the technique relies on R\'enyi DP~\citep{mironov2017renyi} and subsequent works~\citep{balle2018privacy, wang2019subsampled}, which can result in a lossy representation of privacy loss for various mechanisms. An alternative approach involves composing $(\epsilon, \delta)$-DP guarantees using numerical methods like the fast Fourier transform (FFT)~\citep{koskela2020computing, gopi2021numerical}. While effective, this method can be computationally demanding, particularly when a large number of algorithms need to be composed, a common requirement in the training of deep neural networks.

This paper seeks to develop computationally efficient lower and upper privacy bounds, accompanied by accurate estimates, for the composition of private algorithms with finite-sample guarantees\footnote{In this context, ``sample'' refers to the number of compositions of DP algorithms. Henceforth, ``finite-sample'' implies that the bound is non-asymptotic with respect to the number of compositions.}. We utilize a novel privacy definition called $f$-DP~\citep{dong2019gaussian}, which provides a \textit{lossless} interpretation of DP guarantees via a hypothesis testing framework, initially introduced in~\cite{kairouz2015composition}. This allows for precise tracking of privacy loss under composition through a specific operation between functional privacy parameters. Additionally, \cite{dong2019gaussian} devised an approximation tool for evaluating overall privacy guarantees using the central limit theorem (CLT). This tool enables \textit{approximate} $(\epsilon, \delta)$-DP guarantees by leveraging the duality between $(\epsilon, \delta)$-DP and Gaussian DP (GDP, a subfamily of $f$-DP)~\citep{dong2019gaussian}. Although these $(\epsilon, \delta)$-DP guarantees are asymptotically accurate, a usable finite-sample guarantee is currently lacking in the $f$-DP framework.

In this paper, we introduce the \textit{Edgeworth Accountant} as an analytically efficient approach to obtaining finite-sample $(\epsilon, \delta)$-DP guarantees by leveraging the $f$-DP framework. The Edgeworth Accountant employs the Edgeworth approximation~\citep{hall2013bootstrap}, a refined version of the CLT with an improved rate of convergence, to approximate the distribution of certain random variables, namely the privacy-loss log-likelihood ratios (PLLRs). Utilizing a Berry--Esseen type bound developed for the Edgeworth approximation, we establish non-asymptotic upper and lower privacy bounds suitable for applications in privacy-preserving deep learning and federated analytics. Our Edgeworth Accountant's approach is compared at a high level with the GDP approximation as illustrated in Figure~\ref{fig:illustration_comparison}. While the rate of the Edgeworth approximation is relatively straightforward, developing explicit finite-sample error bounds is a complex task. To the best of our knowledge, this is the first instance of such a bound being established in the area of DP, making it an independent contribution in its own right.

To meet diverse practical requirements, we have developed two variants of the Edgeworth Accountant: the Approximate Edgeworth Accountant (AEA) and the Exact Edgeworth Accountant Interval (EEAI). The AEA delivers an asymptotically consistent approximation of the cumulative privacy loss under composition. While the development of this technique is based on $f$-DP and the Edgeworth expansion, this estimate takes a closed-form expression that relates $\epsilon$ to $\delta$, making it computationally more efficient than previous methods \citep{koskela2020computing, gopi2021numerical}. Additionally, the precision of this approach can be enhanced with higher-order Edgeworth expansions if needed. Our experiments demonstrate that the AEA exhibits high accuracy for practical problems in private deep learning and federated analytics. We also introduce an extension of the AEA, named the EEAI, which offers finite-sample lower and upper bounds on the true privacy loss. The EEAI can be used in conjunction with the AEA to provide insights into the estimation error in the overall privacy loss under composition.

A key strength of our accounting algorithm is its exceptional computational efficiency. For composing $m$ identical mechanisms, it achieves computation of privacy loss in constant time, $O(1)$. In more complex scenarios involving $m$ heterogeneous algorithms, its runtime is linear, $O(m)$. This efficiency contrasts with the FFT-based algorithms \citep{gopi2021numerical}, which, despite providing accurate finite-sample bounds, are limited to polynomial runtime in general compositions. Such suboptimal time complexity in FFT-based methods demands extensive resources, particularly as the value of $m$, representing the number of iterations or rounds, often becomes considerably large in deep learning and federated learning. Furthermore, in practical applications where datasets are adaptively used or shared for multiple tasks, a comprehensive tracking of privacy costs for each iteration becomes crucial, thereby increasing the composition count. Our EEAI stands out as the first DP accountant method to offer finite-sample guarantees, optimal time efficiency, and high accuracy for large-scale compositions, making it a valuable addition to the existing toolkit in the research of DP.

\begin{figure}
\begin{equation*}
\begin{tikzpicture}
\node (P0) at (-7cm, 1cm) {\large $f$-DP};
\node (P00) at (-6.8cm, 0.83cm) {};
\node (P1) at (0.5cm, 1cm) {\large \textit{approximate} GDP} ;
\node (P2) at (-7cm, -0.5cm) {\large  
 $(\epsilon, \delta(\epsilon))$-DP
};
\node (P22) at (-6.8cm, -0.27cm) {};
\node (P3) at (5.5cm, 1cm) {\large $(\widehat{\epsilon}_{\text{CLT}}, \widehat{\delta}_{\text{CLT}})$-DP};
\node (P4) at (5.5cm, -0.5cm) {\large $(\widehat{\epsilon}_{\text{EW}}, \widehat{\delta}_{\text{EW}})$-DP};
\getangle(P0)(P1)a;
\getangle(P0)(P2)b;
\draw
(P0) edge[->,>=angle 90] node[above, rotate=\anglea] {\footnotesize
 lossy approximation via CLT} (P1)
(P0) edge[->,>=angle 90] node[below,rotate=\angleb] {\footnotesize PLLR} (P2)
(P22) edge[->,dashed] node{} (P00)
(P1) edge[->,>=angle 90] node[above] {\footnotesize  duality} (P3)
(P2) edge[->,>=angle 90] node[below] {\footnotesize  estimate \& finite-sample bound via Edgeworth approximation} (P4);
\end{tikzpicture}
\end{equation*}
\centering
\caption{The comparison between the GDP approximation in~\cite{dong2019gaussian}, and our Edgeworth Accountant. Both methods start from the exact composition using $f$-DP. Upper:~\cite{dong2019gaussian} uses a CLT type approximation to get a GDP approximation to the $f$-DP guarantee, then converts it to $(\epsilon, \delta)$-DP via duality (Fact \ref{fact:1}). Lower: We losslessly convert the $f$-DP guarantee to an exact $(\epsilon, \delta(\epsilon))$-DP guarantee, with $\delta(\epsilon)$ defined with PLLRs in \eqref{eq:simplify_eps_delta}, and then take the Edgeworth approximation to numerically compute the $(\epsilon, \delta)$-DP. 
}
\label{fig:illustration_comparison}
\end{figure}

The paper is organized as follows. We briefly summarize related work in privacy accounting of DP algorithms as well as our contributions in the remainder of Section \ref{sec:introduction}. In Section \ref{sec:problem setup} we introduce the concept of $f$-DP and its basic properties. We then introduce the notion of privacy-loss log-likelihood ratios in Section \ref{sec:PLLRs} and establish how to use them for privacy accounting based on distribution function approximation. In Section \ref{sec:edgeworth_approximation_bound} we provide a new method, Edgeworth Accountant, that can efficiently and accurately evaluate the privacy guarantees, while providing finite-sample error bounds. Simulation results and conclusions can be found in Sections \ref{sec:experiments} and \ref{sec:discussion}. Proofs and technical details are deferred to the appendices.


\subsection{Motivating applications}\label{sec:motivating_applications}
We now discuss two motivating applications: NoisySGD \citep{song2013stochastic,abadi2016deep,bu2019deep} as well as federated analytics and federated learning \citep{googleblog2020, wang2021federated}. The analysis of DP guarantees of these applications is important yet especially challenging due to the \textit{large} number of compositions involved. Our goal is primarily to devise a general tool to analyze the DP guarantees for these applications.

\noindent\textbf{NoisySGD.} NoisySGD is one of the most popular algorithms for training private deep neural networks. In contrast to the standard SGD, NoisySGD has two additional steps in each iteration: clipping (to bound the sensitivity of the gradients) and noise addition (to guarantee the privacy of models). The details of the NoisySGD algorithm are described in Algorithm \ref{alg:dpsgd} in \Cref{app:algo}.

\noindent\textbf{Federated analytics.}
Federated analytics is a distributed analytical model, which performs statistical tasks through the interaction between a central server and local devices. To complete a global analytical task, in each iteration, the central server randomly selects a subset of devices to carry out local analytics and then aggregates results for the statistical analysis. The total number of iterations is usually very large\footnote{The number of iterations can be small for a single analytical task. However, in most practical cases, many statistical tasks are performed on the same base of users which leads to a large number of total iterations.} in federated analytics, requiring a tight analysis of its DP guarantee.

\subsection{Related work}

\noindent \textbf{Moments accountant and R\'enyi DP.} \cite{abadi2016deep} proposed moments accountant that uses R\'enyi DP~\citep{mironov2017renyi} to give an upper bound for the DP guarantee of composition of private algorithms.
With the help of moments accountant, \cite{abadi2016deep} proposed the DP stochastic gradient descent (DP-SGD) algorithm, whose privacy loss can be bounded effectively. However, as mentioned before, R\'enyi DP can only yield lossy conversion to $(\epsilon, \delta)$-DP, making the upper bound often impractical to use. The runtime of the accountant is independent of $m$, the number of compositions, for DP-SGD, and is $O(m)$ for the composition of general algorithms.

\noindent \textbf{Numerical composition via FFT.} Another line of work~\citep{koskela2020computing, gopi2021numerical} approximated the privacy loss of compositions using the FFT to the convolutions of privacy-loss random variables (PRVs). This notion is closely related to our definition of PLLRs. Though both definitions allow for computing compositions via understanding convolutions of random variables, we note that the two concepts stem from a different analysis framework. Specifically, PRV amounts to finding a pair of random variables that re-parameterizes the privacy curve, which is dual to the trade-off function. On the other hand, PLLRs are defined naturally from the $f$-DP's hypothesis testing perspective, hence the random variables have a direct decomposition into sum of the log-likelihood ratios. As a result, our Proposition \ref{prop:eps_delta_f_dual} is a strict generalization which encompasses their Theorem 3.2 as a special case when $m = 1$. Note that their FFT accountant is the first algorithm that can approximate the privacy loss up to arbitrary precision, and the runtime of their algorithm is $\widetilde{O}(\sqrt{m})$ for DP-SGD and $O(m^{2.5})$ for general compositions.

\noindent \textbf{Analytical composition via characteristic functions.} Recently, \cite{zhu2021optimal} proposed using characteristic functions to analytically compute composition of privacy algorithms. Their algorithm, Analytical Fourier Accountant, yields tight privacy accounting but fails to perform efficient computation for the sub-sampled mechanisms. Their time complexity is $O(1)$ if the characteristic function of their dominating PLD of $m$-fold is simple enough for closed-form composition, and is at least $\Omega(m^2)$ when no closed-form solution was available.

\noindent \textbf{$f$-DP accountant via Edgeworth expansion.} It is worth mentioning that \cite{zheng2020sharp} also used Edgeworth expansion for DP guarantees. 
Specifically, they used Edgeworth approximation as a refinement to the CLT to better approximate the $f$-DP trade-off curve. The most important difference between the two approaches is that we provide a finite-sample error bound that allows for an
exact DP accountant, while they focus solely on an asymptotic approximation to the
trade-off curves. Also, we use Edgeworth approximation on PLLRs to get an estimate of the exact characterization of ($\epsilon$, $\delta$)-DP (the lower path in Figure \ref{fig:illustration_comparison}), while they directly approximate trade-off function $f$ (the same as GDP, using the upper path in Figure \ref{fig:illustration_comparison}). Therefore,  we focus more on the practical side (finite-sample guarantee), and interpretability (directly deal with
($\epsilon$, $\delta$)-DP).

We summarize the comparison of performance of different accountants in Table \ref{table:table1}. Specifically, we focus on their theoretical guarantees and the runtime complexity\footnote{We emphasize that all the $O$ notation in the runtime complexity is up to a logarithmic term of $m$, this is primarily to account for the numerical error in modern computers, and for the fact that expressing the number $m$ itself already requires a time complexity of $O(\log m)$. A detailed analysis of runtime complexity is deferred to \Cref{app:error_analysis}.} when the number of compositions is $m$. We also demonstrate the performance of AEA, GDP, FFT, and MA for small number of compositions in Figure \ref{fig:exact_compare}. Here, we consider the case of composing $m = 1500$ Gaussian mechanisms with $\sigma = 80$. 

\begin{figure}
\centering
\label{fig:exact_compare}
\includegraphics[width=0.5\linewidth]{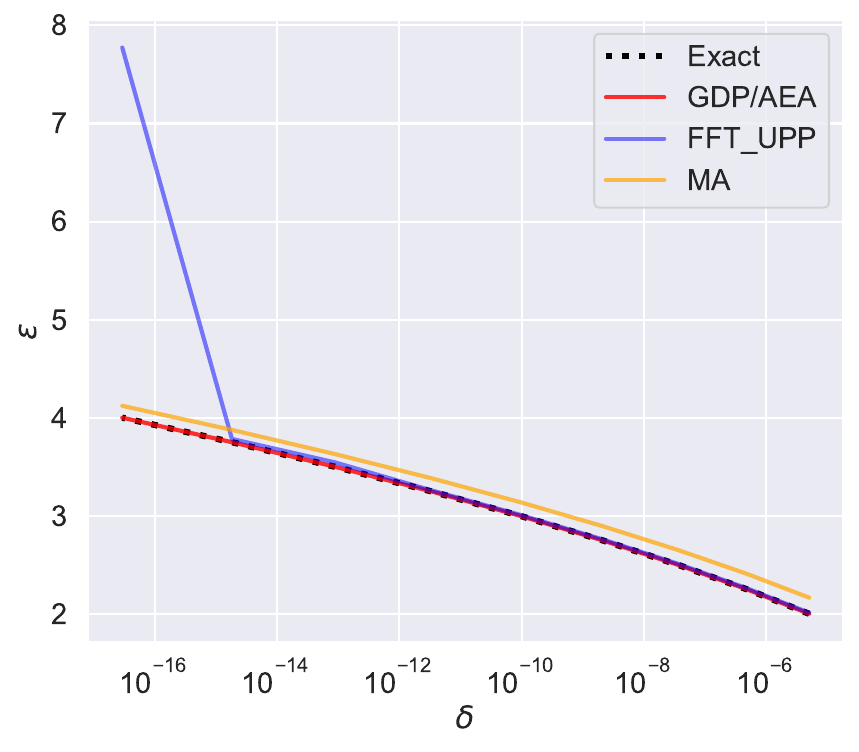}
\caption{Results of estimating/bounding the privacy parameter $\varepsilon(\delta)$ for the composition of $1500$ Gaussian mechanisms with $\sigma = 80$. Curves of ``Exact'' and ``GDP/AEA'' overlap with each other as the latter is exact for composed Gaussian mechanisms. The curve of ``FFT\_UPP''  is close to the ``Exact'' curve for most values of $\delta$, yet exploded when $\delta$ is very small.}
\end{figure}
\begin{table*}
[!htb]{
\footnotesize
	\centering	\begin{tabular}{c|c|c|c}
		\hline
		{Method}& {Finite-sample guarantee} & {Tightness of guarantee} & {Computational complexity}
		\\\hhline{=|=|=|=}
	    GDP/GDP-E &No&N/A&$O(1),~ O(m)$
		\\\hline
		MA & Only upper bound& Loose conversion to $(\epsilon, \delta)$-DP & $O(1), ~O(m)$
		\\\hline
		FFT &Yes&Yes&$O(\sqrt{m}), ~O(m^{2.5})$
		\\\hline
		\textbf{EA} &Yes&Yes$^*$& $O(1), ~O(m)$
		\\\hline
	\end{tabular}
	\caption{ Comparison among different DP accountants. 
	Each entry in the computation complexity contains two columns: (Left) the runtime for the  composition of $m$ identical algorithms; (Right) the runtime for the composition of $m$ general algorithms. 
	GDP: the Gaussian differential privacy accountant~\citep{dong2019gaussian}; GDP-E: the Edgeworth refinement to the GDP accountant~\citep{zheng2020sharp}; MA: the moments accountant using R\'enyi-DP~\citep{abadi2016deep}; FFT: the fast Fourier transform accountant for privacy random variables~\citep{gopi2021numerical}; \textbf{EA}: the Edgeworth Accountant we propose, including both the AEA (Definition \ref{def:AEA}), and the EEAI (Definition \ref{def:EEAI}). The guarantee of EA is tight when the order of the Edgeworth expansion $k$ is high, or when $m$ is large for $k = 1$. }
	\label{table:table1}
}
\end{table*}

\subsection{Our contributions}\label{sec:contribution}
We now briefly summarize our three main contributions.
\par
\noindent \textbf{Improved time-complexity and estimation accuracy.} We propose a new DP accountant method, termed Edgeworth Accountant, which gives finite-sample error bounds in constant/linear time complexity for the composition of identical/general mechanisms. In practice, our method outperforms GDP and moments accountant, with almost the same runtime.

    \noindent \textbf{A unified framework for efficient and computable evaluation of $f$-DP guarantee.} Though the evaluation of $f$-DP guarantee is $\#$P-hard~\citep{murtagh2016complexity}, we provide a general framework to efficiently approximate it in certain important applications. Leveraging this framework, any approximation scheme to the CDFs of the sum of PLLRs can directly transform to a new DP accountant.

    \noindent \textbf{Effcient DP parameter estimation with better accuracy.} 
    Our approximate Edgeworth Accountant is computationally efficient for estimating privacy parameters conveniently under any number of compositions. Compared to moments accountant, AEA can give an estimate that is accurate asymptotically. By leveraging Edgeworth expansion, AEA achieves better accuracy compared to GDP-based methods and the performance can be further improved with higher order expansions.  
    
    \noindent \textbf{Exact finite-sample Edgeworth bound.} To the best of our knowledge, we are the first to use the Edgeworth expansion with finite-sample bounds in the literature of DP. The analysis of the finite-sample bound of the Edgeworth expansion is of its own interest, and has many potential applications. We further derive an explicit adaptive exponential decaying bound for the Edgeworth expansion of the PLLRs, which is the first such result for the Edgeworth expansion.

\section{Preliminaries and Problem Setup}
\label{sec:problem setup}

In this section, we first define the notion DP and its extension $f$-DP.  We then set up the problem  by revisiting our motivating applications. 

A DP algorithm promises that an adversary with perfect information about the entire private dataset in use -- except for a single individual -- would find it hard to distinguish between its presence or absence based on the output of the algorithm~\citep{dwork2006calibrating}. Formally, for $\epsilon > 0 $, and $0 \le \delta < 1$, we consider a (randomized) algorithm $M$ that takes as input a dataset.
\begin{definition}\label{def:DP}
A randomized algorithm $M$ is $ (\epsilon, \delta)$-DP if for any neighboring dataset $ S, S^{\prime} $ differing by an arbitrary sample, and for any (measurable) event $E$, $\mathbb{P}[M(S) \in E] \leqslant \mathrm{e}^{\epsilon} \cdot \mathbb{P}\left[M\left(S^{\prime}\right) \in E\right] + \delta.$
\end{definition}
 In \cite{dong2019gaussian}, the authors propose to use the trade-off between the type-I error and type-II error in place of a few privacy parameters in $(\epsilon, \delta)$-DP. To formally define this new privacy notion, we denote by $P$ and $Q$ as the distribution of $M(S)$ and $M(S^\prime)$ for any two neighboring datasets $S$ and $S^\prime$ and a randomized algorithm $M$.  Let $\phi$ be a (possibly randomized) rejection rule for a hypothesis testing, where $H_0: P  \text{  vs.  }H_1: Q$.  The trade-off function $f$ between $P$ and $Q$ is then defined as the mapping between type-I error to type-II error, that is, $f = T(P, Q):
\alpha  \mapsto \inf _{\phi}\left\{1-\mathbb{E}_{Q}[\phi]: \mathbb{E}_{P}[\phi] \leqslant \alpha\right\}.$ 
This motivates the following definition.
\begin{definition}
A (randomized) algorithm $M$ is $f$-DP if
 $T\left(M(S), M\left(S^{\prime}\right)\right) \geqslant f$
for all neighboring datasets $S$ and $S^{\prime}$.
\end{definition}
We note that the function $f$ is any valid trade-off function $f: [0, 1] \to [0, 1]$ that is defined in ~\cite{dong2019gaussian}. Following the convention of $f$-DP, without loss of generality, we assume that $f$ is a symmetric trade-off function (see Proposition 2 in \cite{dong2019gaussian})\footnote{Note that for ease of read, we always present the final privacy guarantee by a symmetric trade-off function $f$. Therefore, Fact \ref{fact:1} is stated only for $\epsilon\ge 0$. We note as a technical remark that, however, the intermediate ordered pairs of hypotheses may not always be symmetric; for such, the associated $\delta(\epsilon)$ is naturally indexed by $\epsilon\in \mathbb{R}$.}. We highlight the following facts about $f$-DP have been established in~\cite{bu2019deep, dong2019gaussian}.

\begin{fact}[Duality to $(\epsilon, \delta)$-DP]\label{fact:1}
A mechanism is $f$-DP if and only if it is $(\epsilon, \delta(\epsilon))$-DP for all $\epsilon\ge0$, with $\delta(\epsilon) = 1 + f^*(-e^\epsilon)$. Here $g^*(y) = \sup_{-\infty<x<\infty} yx - g(x)$ is the convex conjugate of $g$. 
\end{fact}

\begin{fact}[Composition]
Letting $M_{1}$ and $M_{2}$ be two mechanisms, we define their composition algorithm $M$ as $M(S)=\left(M_{1}(S), M_{2}\left(S, M_{1}(S)\right)\right)$.  In general, the composition of more than two algorithms follows recursively.  Given trade-off functions $f=T(P, Q)$ and $g=T\left(P^{\prime}, Q^{\prime}\right)$, let $f \otimes g=T\left(P \times P^{\prime}, Q \times Q^{\prime}\right)$. Let $m$ be a positive integer, and let $M_1$, . . . , $M_m$ be $m$ algorithms, such that each algorithm $M_t$ is $f_t$-DP. Then the composition theorem states that their $m$-fold composition algorithm is $f_1 \otimes \cdots \otimes f_M$-DP, which is tight in general.
\end{fact}

\begin{fact}[Subsampling]
Consider the following two most common subsampling schemes: (1) (Poisson subsampling) for each individual in the dataset $S$, includes its datum in the subsample independently with probability $p$; (2) (Uniform subsampling) draws a subsample of $S$ that is chosen uniformly at random among all the subsets of $S$ that have cardinal $s := |S| p$, for
some $p \in (0, 1)$.
Denote $\operatorname{Id}(\alpha)=1-\alpha$, and suppose an algorithm $M$ is $f$-DP. The subsampling theorem for $f$-DP states that the Poisson subsampled and uniform subsampled algorithms are both $C_p(f)$-DP, where $C_p(f) :=\min\{f_p, f_p^{-1}\}^{**}$ and $f_{p} = p f+(1-p)\operatorname{Id}$. Here $f_p^{-1}$ denotes the inverse of $f_p$. By definition, note that $C_p(f)$ is a symmetric trade-off function, which is necessary to ensure that the trade-off function of the subsampled algorithm is symmetric.
\end{fact}

\begin{fact}[GDP]\label{fact:4}
To deal with the composition of $f$-DP guarantees,~\cite{dong2019gaussian} introduce the concept of $\mu$-GDP, which is a special case of $f$-DP with $f = G_\mu = T(\mathcal{N}(0, 1), \mathcal{N}(\mu, 1))$. They prove that when all the $f$-DP guarantees are close to the identity, their composition is asymptotically a $\mu$-GDP with some computable $\mu$, which can then be converted to $(\epsilon, \delta)$-DP via duality. However, it comes without a finite-sample bound. 
\end{fact}

\noindent With these facts, we can characterize the $f$-DP guarantee for motivating applications in \Cref{sec:motivating_applications}. 

\noindent\textbf{NoisySGD.} For a NoisySGD with $m$ iterations, subsampling ratio of $p$, and noise multiplier $\sigma$, it is $\min\{f, f^{-1}\}^{**}$-DP~\citep{bu2019deep,dong2019gaussian}, with $f=\left(p G_{1/\sigma}+(1-p) \mathrm{Id}\right)^{\otimes m}$.

\noindent\textbf{Federated analytics.}
Suppose there are $m$ tasks, and each task is $f_i$-DP with $f_i = T(P_i, Q_i)$. Then the overall DP guarantee is characterized by $\bigotimes_{i = 1}^m f_i$-DP. 

\noindent It is easy to see that the $f$-DP guarantee of NoisySGD is a special case of the $f$-DP guarantee of federated analytics with each trade-off function being $f_i = \min\{f_p, f_p^{-1}\}^{**}$, that is, with identical composition of subsampled Gaussian mechanisms. Therefore, our goal is to efficiently and accurately evaluate the privacy guarantee of the general $\bigotimes_{i = 1}^m f_i$-DP with an explicit finite-sample error bound.

\section{Privacy-Loss Log-likelihood Ratios}\label{sec:PLLRs}

We aim to compute the explicit DP guarantees for general composition of trade-off functions of the form $    f = \bigotimes_{i = 1}^m f_i.$
For the $i$-th composition, the trade-off function $f_i = T(P_i, Q_i)$ is realized by the two hypotheses:
$H_{0, i}: w_i\sim P_i \text { vs. } H_{1, i}: w_i\sim Q_i, $
where $P_i, Q_i$ are two distributions and $w_i$ denotes the random variable which is the outcome from the $i$-th algorithm. To evaluate the trade-off function $f = \bigotimes_{i = 1}^m f_i$, we are essentially distinguishing between the two composite hypotheses $H_0: \bm{w} \sim P_1\times P_2 \times \cdots \times P_m ~\text{ vs. }~ H_1: \bm{w} \sim Q_1\times Q_2 \times \cdots \times Q_m$,
where $\bm{w} = (w_1, ..., w_m)$ is the concatenation of all $w_i$'s.  Motivated by the optimal test asserted by the Neyman-Pearson Lemma, we give the following definition.
\begin{definition}\label{def:privacy_loss_loglikelihood}
The associated pair of \textit{privacy-loss log-likelihood ratios (PLLRs)} for two hypotheses $H_0$ and $H_1$ is defined as the logarithm of the Radon--Nikodym derivative of the output distribution under the null hypothesis ($P$) and the output distribution under the alternative hypothesis ($Q$), respectively. Specifically, we can express PLLRs with respect to $H_{0, i}$ and $H_{1, i}$ as
$
X_i \equiv \log \left(\frac{dQ_i(\xi_i)}{dP_i(\xi_i)}\right)
, Y_i \equiv \log \left(\frac{dQ_i(\zeta_i)}{dP_i(\zeta_i)}\right),
$
where $\xi_i \sim P_i, \zeta_i \sim Q_i$.\footnote{For completeness, we explicitly require that all $\xi_i$ and $\zeta_i$ be independent.}
\end{definition}

Note that the definition of PLLRs depends solely on the two hypotheses. Importantly, PLLRs are equivalent to privacy loss random variables \citep{balle2018improving, zhu2021optimal, gopi2021numerical}; however, we prefer the former due to its facility in losslessly converting the $f$-DP guarantee to a collection of $(\epsilon,\delta)$-DP guarantees. The following proposition characterizes the relationship between $\epsilon$ and $\delta$ in terms of the distribution functions of PLLRs.



\begin{proposition}
\label{prop:eps_delta_f_dual}
Let $X_{1}, \ldots, X_{m}$ and $Y_{1}, \ldots, Y_{m}$ be the PLLRs defined above. Let $F_{X, m}, F_{Y, m}$ be the CDFs of $X_{1}+\cdots+X_{m}$ and $Y_{1}+\cdots+Y_{m}$, respectively. Define for $\epsilon\in\mathbb{R}$
\begin{equation}
\label{eq:simplify_eps_delta}
    \delta_{X,Y}(\epsilon) := 1-F_{Y,m}(\epsilon)-e^{\epsilon}(1-F_{X,m}(\epsilon)).
\end{equation}
This function fully characterizes the trade-off function associated with the ordered pair of composite hypotheses  $(P_1\times P_2 \times \cdots \times P_m,Q_1\times Q_2 \times \cdots \times Q_m)$.
In particular, when the associated trade-off function $f = \bigotimes_{i = 1}^m f_i$ is symmetric, then the composed mechanism is $f$-DP if and only if it is $(\epsilon, \delta_{X,Y}(\epsilon))$-DP for all $\epsilon\ge 0$.
\end{proposition}

For the complete proof of Proposition \ref{prop:eps_delta_f_dual}, please refer to Section \ref{app:proof_sec_C_1}. 

Proposition \ref{prop:eps_delta_f_dual} expresses the function $\delta_{X,Y}(\epsilon)$ associated with the ordered pair of composite hypotheses in terms of the distribution functions of PLLRs. This reflects the primal-dual relationship between trade-off functions and collections of $(\epsilon, \delta(\epsilon))$-DP guarantees. When the associated trade-off function is symmetric, Fact \ref{fact:1} recovers the usual equivalence between $f$-DP and $(\epsilon, \delta_{X,Y}(\epsilon))$-DP using only $\epsilon\ge0$.

Note that related formulations in terms of privacy loss random variables appear in \cite{balle2018improving, zhu2021optimal, gopi2021numerical}. Our Proposition \ref{prop:eps_delta_f_dual} is particularly useful as it reduces the $\#$P-complete $f$-DP accounting problem to approximating $F_{X,m}$ and $F_{Y,m}$ in (\ref{eq:simplify_eps_delta}). Although the existence of approximation schemes for DP composition has been established in prior work \citep{murtagh2015complexity}, the primary challenge remains in developing tractable and accurate approximation methods. Our key contribution lies in demonstrating that the Edgeworth expansion provides a principled and analytically tractable framework for approximating the CDFs $F_{Y,m}$ and $F_{X,m}$, thereby yielding computationally efficient, accurate, robust, and finite-sample privacy accountants.


Of note, the above relationship is general in the sense that we make no assumption on the private mechanisms. 

Definition \ref{def:privacy_loss_loglikelihood} can be applied directly  when  $\frac{dQ_i(\xi_i)}{dP_i(\xi_i)}$ is easy to compute, which corresponds to the case without subsampling. To deal with the case with subsampling, one must take into account that the subsampled DP guarantee is the double conjugate of the minimum of two asymmetric trade-off functions (for example, recall the trade-off function of a single sub-sampled Gaussian mechanism is $\min\{f_p, f_p^{-1}\}^{**}$, where $f_p = (pG_{1/\sigma} + (1-p)\mathrm{Id})$). In general, the composition of multiple subsampled mechanisms satisfies $f$-DP for $f = \min\{\otimes_{i = 1}^m f_{i, p_i}, \otimes_{i = 1}^m f_{i, p_i}^{-1}\}^{**}$. This general form makes the direct computation of the PLLRs through composite hypotheses infeasible, as it is hard to write $f$ as a trade-off function for some explicit pair of hypotheses $H_0$ and $H_1.$ 
Therefore, instead of using one single sequence of PLLRs directly corresponding to $f$, we shall use a family of sequences of PLLRs.
In general,
suppose we have a mechanism characterized by some $f$-DP guarantee, where $f = \left(\inf_{\alpha\in \mathcal{I}}\{f^{(\alpha)}\}\right)^{**}$, for index set $\mathcal{I}$. That is, $f$ is the tightest possible trade-off function satisfying all the $f^{(\alpha)}$-DP.  Suppose further that for each $\alpha$, we can find a sequence of computable PLLRs corresponding to $f^{(\alpha)}$, which allows us to obtain a collection of $(\epsilon, \delta^{(\alpha)}(\epsilon))$-DP guarantees. 

\begin{lemma}\label{lem:composite_duality}
Suppose that for each $\alpha\in \mathcal{I}$, the function $f^{(\alpha)}$ and the function $\delta^{(\alpha)}:\mathbb{R}\to[0, 1]$ satisfy that a mechanism is $f^{(\alpha)}$-DP if and only if it is $(\epsilon, \delta^{(\alpha)}(\epsilon))$-DP for all $\epsilon$. Let $f := \left(\inf_{\alpha\in \mathcal{I}}\{f^{(\alpha)}\}\right)^{**}$, and $\delta(\epsilon):=\sup_\alpha\{\delta^{(\alpha)}(\epsilon)\}$. Then a mechanism is $f$-DP if and only if it is $(\epsilon, \delta(\epsilon))$-DP for all $\epsilon\in\mathbb{R}$. In particular, when $f$ is symmetric, a mechanism is $f$-DP if and only if it is $(\epsilon, \delta(\epsilon))$-DP for all $\epsilon\ge 0$.
\end{lemma}
We defer the proof of this lemma to the appendices. The intuition is that both $\left(\inf_{\alpha\in \mathcal{I}}\{f^{(\alpha)}\}\right)^{**}$-DP and $(\epsilon, \sup_\alpha\{\delta^{(\alpha)}(\epsilon)\})$-DP correspond to the \textit{tightest} possible DP-guarantee for the entire collection. 

Lemma \ref{lem:composite_duality} allows us to characterize the subsampled Gaussian mechanism using two sequences of PLLRs. As mentioned above, it is $f$-DP with $f = \min\{\otimes_{i = 1}^m f_{i, p}, \otimes_{i = 1}^m f_{i, p}^{-1}\}^{**}$, where each $f_{i, p} = \left(p G_{1/\sigma}+(1-p) \mathrm{Id}\right)$. For the first part, the PLLRs corresponding to $\otimes_{i = 1}^m f_{i, p} = \left(p G_{1/\sigma}+(1-p) \mathrm{Id}\right)^{\otimes m}$ are given by $X_i^{(1)} = \log(1-p+pe^{\mu\xi_i-\frac{1}{2}\mu^2}),$ and $ Y_i^{(1)} = \log(1-p+pe^{\mu\zeta_i-\frac{1}{2}\mu^2}),$ for $1\le i \le m$, with $\xi_i\sim \mathcal{N}(0, 1)$, $\zeta_i\sim p\mathcal{N}(0, 1) + (1-p) \mathcal{N}(\mu,1).$ And for the second part,
the PLLRs corresponding to $\otimes_{i = 1}^m f_{i, p}^{-1} = \left((p G_{1/\sigma}+(1-p) \mathrm{Id})^{-1}\right)^{\otimes m}$ are given by 
    $X_i^{(2)} = -\log(1-p+pe^{\mu\zeta_i-\frac{1}{2}\mu^2}),$ and $Y_i^{(2)} = -\log(1-p+pe^{\mu\xi_i-\frac{1}{2}\mu^2}),$
for $1\le i \le m$, with $\xi_i\sim \mathcal{N}(0, 1)$, $\zeta_i\sim p\mathcal{N}(0, 1) + (1-p) \mathcal{N}(\mu,1).$
Now substituting $F_{X^{(1)}, m}$ and $F_{Y^{(1)}, m}$ by any approximation (for example, using the CLT or Edgeworth), we get a computable relationship in terms of the $(\epsilon, \delta^{(1)}(\epsilon))$-DP; and similarly, we can get a relationship in terms of the $(\epsilon, \delta^{(2)}(\epsilon))$-DP. We conclude that the subsampled Gaussian mechanism is  $(\epsilon, \max\{\delta^{(1)}(\epsilon), \delta^{(2)}(\epsilon)\})$-DP. 
Here we assume we know the form of $X_i^{(\alpha)}$ and $Y_i^{(\alpha)}$ corresponding to $f_i^{(\alpha)}$. This is also assumed for accounting algorithms in \cite{gopi2021numerical, zhu2021optimal} as they require the form of PLLRs or the analytical form of their characteristic functions. We can relax this assumption to only knowing its first fourth moments in Lemma \ref{lem:first_order_edgeworth_bound}.



\subsection{Transferred error bound based on CDF approximations}
\label{sec:eps_bound}
As discussed above, Lemma \ref{lem:composite_duality} allows us to characterize the double conjugate of the infimum of a collection of $f^{(\alpha)}$-DPs via analyzing each sequence of PLLRs separately. As a result, our focus is to compute the bounds of $\delta^{(\alpha)}$ for each single trade-off function $f^{(\alpha)}$.  To fulfill the purpose,
we seek an efficient algorithm for approximating distribution functions of the sum of PLLRs, namely, $F_{X^{(\alpha)},m}, F_{Y^{(\alpha)},m}$. This perspective provides a general framework that naturally encompasses many existing methods, including fast Fourier transform \citep{gopi2021numerical} and the characteristic function method \citep{zhu2021optimal}. 
They can be viewed as different methods for finding upper and lower bounds of $F_{X^{(\alpha)},m}, F_{Y^{(\alpha)},m}$. 
Specifically, we denote the upper and lower bounds of $F_{X^{(\alpha)},m}$ by $F^+_{X^{(\alpha)},m}$ and $F^-_{X^{(\alpha)},m}$, and similarly for $F_{Y^{(\alpha)},m}$.
These bounds can be easily converted to the error bounds on privacy parameters of the form $g^{(\alpha)-}(\epsilon)\le \delta^{(\alpha)}(\epsilon)\le g^{(\alpha)+}(\epsilon)$, for all $\epsilon> 0$, where 
\begin{equation}
\begin{aligned}
     g^{(\alpha)+}(\epsilon) &= 1-F_{Y^{(\alpha)},m}^-(\epsilon)-e^{\epsilon}(1-F_{X^{(\alpha)},m}^+(\epsilon)),\\
     g^{(\alpha)-}(\epsilon) &= 1-F_{Y^{(\alpha)},m}^+(\epsilon)-e^{\epsilon}(1-F_{X^{(\alpha)},m}^-(\epsilon)).
     \label{eq:Edgeworth_upper_lower_bound}
\end{aligned}
\end{equation}
Thus, the DP guarantee of $\left(\inf_\alpha f^{(\alpha)}\right)^{**}$-DP in the form of $(\epsilon, \delta(\epsilon))$ satisfies $
\sup_\alpha \{g^{(\alpha)-}(\epsilon)\}\le \delta(\epsilon) \le \sup_\alpha \{g^{(\alpha)+}(\epsilon)\}$, for all $\epsilon > 0$.
To convert the guarantee of the form $(\epsilon, \delta(\epsilon))$ for all $\epsilon>0$ to the guarantee of the form $(\epsilon(\delta), \delta)$ for all $\delta\in [0, 1)$, we can invert the above bounds on $\delta(\epsilon)$ and obtain the bounds of the form $\epsilon^-(\delta) \le \epsilon(\delta) \le \epsilon^+(\delta)$. Here $\epsilon^+(\delta)$ is the largest root of equation
$\delta = \sup_\alpha\{g^{(\alpha)+}(\cdot)\},$
and $\epsilon^-(\delta)$ is the smallest non-negative root of equation
$    \delta = \sup_\alpha\{g^{(\alpha)-}(\cdot)\}.$ 
\begin{remark}
 In practice, we often need to solve for those roots numerically, and we need to specify a finite range in which we find all the roots. In Appendix 
 \ref{app:algo}, we exemplify how to find such range for NoisySGD, see Remark \ref{rmk:bound_on_g} in the appendices for details. 
\end{remark}

\section{Edgeworth Accountant}
\label{sec:edgeworth_approximation_bound}

In what follows, we present a new approach, \textit{Edgeworth Accountant}, based on the Edgeworth expansion to approximate the distribution functions of the sum of PLLRs. For simplicity, we demonstrate how to obtain the Edgeworth Accountant for any trade-off function $f^{(\alpha)}$ based on a single sequence of PLLRs $\{X_i^{(\alpha)}\}_{i=1}^m, \{Y_i^{(\alpha)}\}_{i=1}^m$. Henceforth, we drop the superscript $\alpha$ when it is clear from the context. Specifically, we derive an approximate Edgeworth Accountant (AEA) and the associated exact Edgeworth Accountant interval (EEAI) for $f$ with PLLRs $\{X_i\}_{i=1}^m$ and $\{Y_i\}_{i=1}^m$. 


\subsection{Approximate Edgeworth Accountant}\label{sec:Edgeworth_Accountant}

To approximate the CDF of a random variable $X = \sum_{i = 1}^mX_i$, we introduce the Edgeworth expansion in its most general form, where $X_i$'s are independent but not necessarily identical. Such generality allows us to account for composition of heterogeneous DP algorithms.
Suppose $\mathbb{E}\left[X_{i}\right]=\mu_i$ and $\gamma_{p, i}:=\mathbb{E}\left[(X_{i}-\mu_i)^{p}\right]<+\infty$ for some $p\ge 4$. We define $B_{m}:=\sqrt{\sum_{i=1}^{m} \mathbb{E}\left[(X_{i}-\mu_i)^{2}\right] }$, and $ \sum_{i=1}^m \mu_i = M_m$. So, the standardized sum can be written as $S_{m}:=(X - M_m) / B_{m}$. We denote $E_{m, k, X}(x)$ to be the $k$-th order Edgeworth approximation of $S_m$. Note that the central limit theorem (CLT) can be viewed as the $0$-th order Edgeworth approximation. The first-order Edgeworth approximation is given by adding one extra order $O(1/\sqrt{m})$ term to the CLT, that is,
$E_{m, 1, X}(x) = \Phi(x)-\frac{\lambda_{3, m}}{6 \sqrt{m}}\left(x^{2}-1\right) \phi(x).$
Here, $\Phi$ and $\phi$ are the CDF and PDF of a standard normal distribution, and $\lambda_{3, m}$ is a  constant to be defined in Lemma \ref{lem:first_order_edgeworth_bound}. 
It is known that the Edgeworth approximation of order $p$ has an error rate of $O(m^{-(p+1) / 2})$ (see, for example, \cite{hall2013bootstrap}). This desirable property motivates us to use the rescaled Edgeworth approximation $G_{m, k, X}(x)=E_{m, k, X}\left((x-M_m)/B_m\right)$ and $G_{m, k, Y}(x)=E_{m, k, Y}\left((x-M_m)/B_m\right)$ to approximate $F_{X, m}(x)$ and $F_{Y,m}(x)$, respectively, in \eqref{eq:simplify_eps_delta}. This is what we term the \textit{approximate Edgeworth Accountant} (AEA).
\begin{definition}[AEA]\label{def:AEA}
The $k$-th order AEA that defines $\delta(\epsilon)$ for $\epsilon>0$ is given by
$    \delta(\epsilon) = 1-G_{m, k, Y}(\epsilon)-e^{\epsilon}(1-G_{m, k, X}(\epsilon))$ for all $\epsilon>0.$
\end{definition}
Notably, the CLT-based approximation can be thought of as a special case corresponding to the zeroth-order Edgeworth expansion ($k = 0$). Therefore, the AEA framework provides a generalization of the GDP accountant, achieving superior convergence rates through the incorporation of higher-order correction terms.

\subsection{Edgeworth Accountant with finite-sample guarantees}\label{sec:exact_Edgeworth_Accountant}
Asymptotically, AEA is an exact accountant, due to the rate of convergence Edgeworth approximation admits. In practice, however, the finite-sample guarantee is still missing since the exact constant of such rate is unknown. To obtain a computable $(\epsilon, \delta(\epsilon))$-DP bound via \eqref{eq:simplify_eps_delta}, we require the finite-sample bounds on the approximation error of the CDF for any finite number of iterations $m$. Suppose that we can provide a finite-sample  bound using Edgeworth approximation of the form
$
    |F_{X,m}(x)-G_{m, k, X}(x)|\le \Delta_{m, k, X}(x),
$
where $\Delta_{m, k, X}(x)$ is computable.
Then  we have
\begin{equation}\label{eq:Edgeworth_approximation_general_bound_form}
\begin{aligned}
    &F_{X,m}^+(x) = G_{m, k, X}(x) + \Delta_{m, k, X}(x), \\
    &F_{X,m}^-(x) = G_{m, k, X}(x) - \Delta_{m, k, X}(x),
    \end{aligned}
\end{equation}
and similarly for $F_{Y,m}$. We now define the \textit{exact Edgeworth Accountant interval} (EEAI).
\begin{definition}[EEAI]\label{def:EEAI}
The $k$-th order EEAI associated with privacy parameter $\delta(\epsilon)$ for $\epsilon>0$ is given by $[\delta^{-},\delta^{+}]$, where for all $\epsilon>0$
\begin{equation}
\label{eq:EEAI}
\begin{aligned}
     \delta^{-}(\epsilon) \equiv & 1-G_{m, k, Y}(\epsilon)-\Delta_{m, k, Y}(\epsilon)
     -e^{\epsilon}(1-G_{m, k,X}(\epsilon)+\Delta_{m, k, X}(\epsilon)),\\
    \delta^{+}(\epsilon) \equiv & 1-G_{m, k, Y}(\epsilon)+\Delta_{m, k, Y}(\epsilon)
    -e^{\epsilon}(1-G_{m, k, X}(\epsilon)-\Delta_{m, k, X}(\epsilon)).
\end{aligned}
\end{equation}
\end{definition}

To bound the EEAI, it suffices to have a finite-sample bound on $\Delta_{m,k, X}(\epsilon)$ and  $\Delta_{m,k, Y}(\epsilon)$.


\subsubsection{Uniform bound on PLLRs}
\label{sec:edge1_bound}
We now deal with the bound of the Edgeworth approximation on PLLRs in \eqref{eq:Edgeworth_approximation_general_bound_form}. 
Our starting point is a uniform bound of the form
$
\Delta_{m, k, X}(x) \le c_{m, k, X},\text{ for all }x.
$ The bound for $\Delta_{m, k, Y}(x)$ follows {identically}.
To achieve this goal, we follow the analysis on the finite-sample bound in~\cite{derumigny2021explicit}.  
We state the bound of the first-order Edgeworth expansion.

\begin{lemma}\label{lem:first_order_edgeworth_bound}
Define the average individual standard deviation $\bar{B}_{m}:=B_{m} / \sqrt{m}$ and the average standardized $r$-th cumulant as $\lambda_{k, m}:=\frac{1}{m} \sum_{j=1}^{m} k_{r, j}/ \bar{B}_{m}^{3}$, where $k_{r,j}$ is the $r$-th centralized cumulant of the $j$-th sample. With bounded moments of order four, that is, $\gamma_{4, i}<+\infty$ for $1\le i \le m$, we have the (uniform) bound on Edgeworth expansion as
\begin{align*}
    &\Delta_{m, 1, X} \leq \frac{0.1995 \widetilde{K}_{3, m}}{\sqrt{m}}\\
    &+\frac{0.031 \widetilde{K}_{3, m}^{2}+0.195 K_{4, m}+0.054\left|\lambda_{3, m}\right| \widetilde{K}_{3, m}+0.038 \lambda_{3, m}^{2}}{m}\\&+r_{1, m},
\end{align*}
where $K_{p, m}:=m^{-1} \sum_{i=1}^{m} \mathbb{E}\left[\left|X_{i}-\mu_i\right|^{p}\right] /\left(\bar{B}_{m}\right)^{p}$, which is the average standardized $p$-th absolute moment, 
and $\widetilde{K}_{3, m}:=K_{3, m}+\frac{1}{m} \sum_{i=1}^{m} \mathbb{E}\left|X_{i}-\mu_i\right| \gamma_{2, i} / \bar{B}_{m}^{3}$. Here $r_{1, m}$ is a remainder term of order $O(1/m^{5/4})$ that depends only on $K_{3, m}, K_{4, m}$ and $\lambda_{3, m}$, and is defined in Equation \eqref{eq:r_1,n} in Appendix~\ref{appendx:e1}.
\end{lemma}
Note that this lemma deals with the first-order Edgeworth approximation which can be generalized to the higher-order Edgeworth approximations. We defer the analysis of those higher-order approximations to future work. The expression of $r_{1,m}$ only involves real integration with known constants which can be numerically evaluated in constant time.
\begin{remark}\label{remark:4.4}
The precision of the EEAI is inherently linked to the finite-sample bound rate of the Edgeworth expansion, with any improvements in higher-order Edgeworth expansion bounds directly applicable to our EEAI by substituting the corresponding $\Delta_{m, k, X}(\epsilon)$. While we primarily demonstrate the application for $k=1$ using the first-order expansion, it is noteworthy that Lemma \ref{lem:first_order_edgeworth_bound} yields a bound of order $O(1/\sqrt{m})$ due to our consideration of general independent, but not necessarily identical, random variables. In contrast, we show in Appendix \ref{appen:details_edgeworth} how an improved $O(1/m)$ rate can be achieved in the i.i.d.\ case. Notably, the $O(1/m)$ bound in Theorem 8 of \cite{dong2019gaussian} is limited to homogeneous composition of pure DP ($\delta = 0$). Our current first-order bound, leveraging Theorem 2 from \cite{derumigny2021explicit}, is particularly effective for large $m$. However, the potential for bounds utilizing higher-order Edgeworth expansions to enhance precision across all $m$ values remains an open avenue for future research.


\end{remark}

\subsubsection{Adaptive exponential decaying bound for NoisySGD}\label{sec:non_uniform_bound_GM}
One specific concern of the bound derived in the previous section is that it is uniform in $\epsilon$. Note that in \eqref{eq:simplify_eps_delta}, there is an amplification factor of error by $e^\epsilon$ in front of $F_{X,m}$. Therefore, as long as $\epsilon$ grows in $m$ with order at least $\epsilon = \Omega(\log{m})$, the error term in \eqref{eq:simplify_eps_delta} scales with order $e^{\Omega({\log m})}/O(m) = \Omega(1)$. 

In this section, we study the compositions of subsampled Gaussian mechanism (including NoisySGD and many federated learning algorithms), where we are able to improve the previous bound when $\epsilon$ is large. Informally, omitting the dependence on $m$, we want to have a bound of the form $|F_{X,m}(\epsilon) - G_{m, k, X}(\epsilon)| = O(e^{-\epsilon^2})$ to offset the effect of $e^\epsilon$ in front of $F_{X,m}$. To this end, we first prove that the tail bound of $F_{X, m}(\epsilon)$ is of order $O(e^{-\epsilon^2})$, with exact constant. Combining with  the tail behavior of the Edgeworth expansion, we conclude that the difference has the desired convergence rate. Following the discussion in \Cref{sec:PLLRs}, we need to calculate the bounds for two sequences of PLLRs separately. Here we focus on the sequence of PLLRs corresponding to $\left(p G_{1/\sigma}+(1-p) \mathrm{Id}\right)^{\otimes m}$. These PLLRs are given by $X_i = \log(1-p+pe^{\mu\xi_i-\frac{1}{2}\mu^2})$, where $\xi_i\sim N(0, 1)$.  
The following theorem characterizes the tail behavior of $F_{X,m}$. The tail bound of the sum of the other sequence of PLLRs corresponding to $((p G_{1/\sigma}+(1-p) \mathrm{Id})^{-1})^{\otimes m}$ has the same rate, and can be proved similarly.

\begin{theorem}
\label{thm:sub_gaussian_bound}
There exist some positive constant $a$, and some associated constant $\eta(a)> 0$, 
such that the tail of $F_{X,m}$ can be bounded as
$
    1 - F_{X, m}(\epsilon) = \mathbb{P}\left(\sum_{i = 1}^m X_i \ge \epsilon
    \right) \le 2\exp\left(-\frac{(\epsilon + m \eta)^2}{8m \tau^2}
    \right),
$
where $\tau^2 = \max\left\{\frac{(\log(1-p+pe^{\mu a-\frac{1}{2}\mu^2})+\mu(a^+-a)-\log(1-p))^2}{4},\right.$ $\left.\mu^2, \frac{(a^+ - a)^2\mu^2}{2\log(\Phi(a^+) - \Phi(a))}\right\}$ and $a^+ = \frac{\phi(a)}{1 - \Phi(a)}$. 
\end{theorem}
The proof of Theorem \ref{thm:sub_gaussian_bound} is deferred to Appendix \ref{app:proof_sec_4} along with its dependent technical lemmas.  From the above theorem, we know that the tail of $F_{X,m}(\epsilon)$  is $O(e^{-\max\{\epsilon^2 /m, m\}}) = o(e^{-\epsilon})$, as long as $\epsilon = o({m})$. Note that in this case, the tail of the rescaled Edgeworth expansion is  of the same order $O(e^{-\max\{\epsilon^2 /m, m\}}) = o(e^{-\epsilon})$. Therefore, we can give a finite-sample bound  of the same rate for the difference between $F_{X,m}(\epsilon)$ and its approximation $G_{m, k, X}$ at large $\epsilon$. 
Note that this finite-sample bound scales better than uniform bound in Lemma \ref{lem:first_order_edgeworth_bound} when $m$ and $\epsilon$ are large. 

\subsection{Implementation of Edgeworth Accountant}\label{app:algo_EA}


To describe how to implement the Edgeworth Accountant, it is necessary to extend the definitions of the AEA and EEAI from specific trade-off functions $f^{(\alpha)}$ to general trade-off functions of the form $\left(\inf_\alpha f^{(\alpha)}\right)^{**}$.

\begin{definition}[AEA for general trade-off function]\label{def:AEA_full}
The $k$-th order AEA of $\left(\inf_\alpha f^{(\alpha)}\right)^{**}$-DP that defines $\delta(\epsilon)$ for $\epsilon>0$ is given by $\delta(\cdot) = \sup_\alpha \delta^{(\alpha)}(\cdot)$, where
\begin{equation}\label{eq:AEA_full}
    \delta^{(\alpha)}(\epsilon) = 1-G_{m, k, Y^{(\alpha)}}(\epsilon)-e^{\epsilon}(1-G_{m, k, X^{(\alpha)}}(\epsilon)),
\end{equation}    
for any $\alpha$.
\end{definition}

\begin{definition}[EEAI for general trade-off function]\label{def:EEAI_full}
The $k$-th order EEAI of $\left(\inf_\alpha f^{(\alpha)}\right)^{**}$-DP associated with privacy parameter $\delta(\epsilon)$ for $\epsilon>0$ is given by $[\delta^{-},\delta^{+}]$, where
$\delta^{-}(\cdot) = \sup_\alpha \delta^{(\alpha)-}(\cdot)$, $\delta^{+}(\cdot) = \sup_\alpha \delta^{(\alpha)+}(\cdot)$, and 
\begin{equation}
\label{eq:EEAI_full}
\begin{aligned}
    & \delta^{(\alpha)-}(\epsilon) \equiv 1-G_{m, k, Y^{(\alpha)}}(\epsilon)-\Delta_{m, k, Y^{(\alpha)}}(\epsilon)-e^{\epsilon}(1-G_{m, k,X^{(\alpha)}}(\epsilon)+\Delta_{m, k, X^{(\alpha)}}(\epsilon)),\\
    & \delta^{(\alpha)+}(\epsilon) \equiv 1-G_{m, k, Y^{(\alpha)}}(\epsilon)+\Delta_{m, k, Y^{(\alpha)}}(\epsilon)-e^{\epsilon}(1-G_{m, k, X^{(\alpha)}}(\epsilon)-\Delta_{m, k, X^{(\alpha)}}(\epsilon)).
\end{aligned}
\end{equation}
\end{definition}

For completeness, we also provide the formal definition of general $k$-th order Edgeworth Expansion $E_{m,k, X}$. More details can be found, for example, in \cite{hall2013bootstrap}.
\begin{definition}[Definition of $k$-th order Edgeworth Expansion]
\label{def:k-th_edgeworth}
For any sequence of $m$ distributions $X_1, ..., X_m$, let $S_m = \frac{\sum_{i=1}^m X_i -\sum_{i=1}^m \mathbb{E} X_i}{\bar{B}_m}$ be the standardized sum, where $\bar{B}_m = \sqrt{\Var{\sum_{i=1}^m X_i}}$. Define $\lambda_r = \frac{1}{m} \frac{\sum_{j=1}^m \kappa_{r, j}}{\bar{B}_m^r}$, where $\kappa_{r, j}$ is the $r$-th cumulant of $X_j$. Assume $X_i$'s have $(k+2)$-th cumulant. We define the $k$-th Edgeworth expansion $E_{m,k, X}$ as 
\begin{equation}
    E_{m, k, X}(x) = \Phi(x) + \sum_{r = 1}^{k} \frac{1}{n^{r/2}} \frac{P_r(-D)}{D} \phi(x),
\end{equation}
where $D$ is the differential operator, and $P_r(-D)$ is a polynomial of degree $3r$. The explicit form of $P_r(-D)$ can be written as
\begin{equation*}
    P_r(-D) = \sum \left(\prod_i \frac{1}{k_i !}\left(\frac{\lambda_{i+2}}{(i + 2)!}\right)^{k_i} (-D)^{k_i(i + 2)}\right),
\end{equation*}
where the summation is over all the integer partitions of $m$ such that $\sum_i ik_i = m$.
\end{definition}

With the definitions above in place, we now present the detailed implementations of both the AEA and EEAI.
\begin{center}
\resizebox{0.99\textwidth}{!}{
\begin{minipage}{\textwidth}
\begin{algorithm}[H]
	\caption{AEA}\label{alg:AEA}

\begin{algorithmic}
\STATE \textbf{Parameters:} $m$ general mechanisms $M_1,..., M_m$, $\epsilon\ge 0$, and order $k\ge 1$.
		\FOR{$i = 1, \ldots, m$}
		\STATE  \hspace{0.5cm}{Analytically encode all the corresponding PLLRs for $M_i$, $\{(X_i^{(\alpha)}, Y_i^{(\alpha)})\}_\alpha$ for all $\alpha$.}
		\STATE  \hspace{0.5cm}{Numerically calculate the cumulants up to order $k + 2$ for $X_i^{(\alpha)}$ and $Y_i^{(\alpha)}$ for all $\alpha$.}		
		\ENDFOR
		\STATE {Calculate $G_{m, k, X^{(\alpha)}}(\epsilon)$ and $G_{m, k, Y^{(\alpha)}}(\epsilon)$ for each $\alpha$ using $k$-th order Edgeworth expansion.}
		\STATE {Calculate $\delta^{(\alpha)}(\epsilon)$ for each $\alpha$ by \eqref{eq:AEA_full}.}
		\STATE {\bf Output} $\sup_\alpha \delta^{(\alpha)}(\epsilon).$
	\end{algorithmic}
\end{algorithm}
\end{minipage}
}
\end{center}

And similarly, we present the algorithm for the general EEAI.

\begin{center}
\resizebox{0.99\textwidth}{!}{
\begin{minipage}{\textwidth}
\begin{algorithm}[H]
	\caption{EEAI}\label{alg:EEAI}
\begin{algorithmic}[0]
\STATE  \textbf{Parameters:} $m$ general mechanisms $M_1,..., M_m$, $\epsilon\ge 0$, and fixed order $k = 1$.
		\FOR{$i = 1, \ldots, m$}
		\STATE \hspace{0.5cm}{Analytically encode all the corresponding PLLRs for $M_i$, $\{(X_i^{(\alpha)}, Y_i^{(\alpha)})\}_\alpha$ for all $\alpha$.}
		\STATE  \hspace{0.5cm}{Numerically calculate the cumulants up to order $4$ for $X_i^{(\alpha)},$ and $Y_i^{(\alpha)}$ for all $\alpha$.}		
		\ENDFOR
		\STATE {Calculate $G_{m, 1, X^{(\alpha)}}(\epsilon)$ and $G_{m, 1, Y^{(\alpha)}}(\epsilon)$ for each $\alpha$ using first order Edgeworth expansion.}
		\STATE {Calculate $\Delta_{m, 1, X^{(\alpha)}}(\epsilon)$ and $\Delta_{m, 1, Y^{(\alpha)}}(\epsilon)$ for each $\alpha$ using Lemma \ref{lem:first_order_edgeworth_bound} or \Cref{thm:sub_gaussian_bound}.}
		\STATE {Calculate $\delta^{(\alpha)+}(\epsilon)$ and $\delta^{(\alpha)-}(\epsilon)$ for each $\alpha$ by \eqref{eq:EEAI_full}.}
		\STATE {\bf Output} $[\sup_\alpha \delta^{(\alpha)-}(\epsilon), \sup_\alpha \delta^{(\alpha)+}(\epsilon)].$
	\end{algorithmic}
\end{algorithm}
\end{minipage}
}
\end{center}

Note that \Cref{alg:AEA} and \Cref{alg:EEAI} seek to find an estimate of $\epsilon$ and bounds on $\delta$ for any given $\epsilon$, respectively. And both algorithms run in constant/linear time for $m$ identical/general compositions. In practice, people often would like to find an estimate or bounds on $\epsilon$ given an $\delta$. To get such an estimate of $\epsilon$ given $\delta$, we can directly inverse the \Cref{alg:AEA}\footnote{Note that if we substitute Edgeworth approximation with the true CDF of PLLRs , it is direct to  show (by taking derivative) that $\delta^{(\alpha)}(\epsilon)$ is always a decreasing function of $\epsilon$, and the supremum  of decreasing functions is still a decreasing function. Therefore, we can always take inversion.}. And to get upper and lower bounds of $\epsilon$ given $\delta$, we can use the inversion method discussed in \Cref{sec:eps_bound}, and specifically, the equations in \eqref{eq:Edgeworth_upper_lower_bound}.

We also supplement these algorithms with a detailed time complexity analysis, which is presented in Appendix \ref{app:error_analysis}.

\subsection{Extension to other mechanisms}\label{sec:non_Gaussian_non_uniform_bound}

Note that our analysis framework is applicable to a wide range of common noise-adding mechanisms. Specifically, Lemma \ref{lem:first_order_edgeworth_bound} only requires the distribution of PLLRs to have bounded fourth moments. And for many common mechanisms, a counterpart of \Cref{thm:sub_gaussian_bound} can be proved similarly. We now demonstrate how to generalize our analysis to the Laplace Mechanism. 

\noindent\textbf{The Laplace mechanism.}
It is straightforward to verify that the trade-off function for subsampled Laplace Mechanisms is given by $\min\{(p L_{\mu} + (1-p)\mathrm{Id})^{\otimes m}, ((p L_{\mu} + (1-p)\mathrm{Id})^{-1})^{\otimes m}\}^{**}$, where $L_{\mu} = T(\text{Lap}(0, 1), \text{Lap}(\mu, 1))$.
 The two associated sequences of PLLRs $X_i$ and $Y_i$ can be expressed as:
$
    X_i^{(1)} \equiv \log \left(1-p+p \mathrm{e}^{|\xi| - |\xi - \mu|}\right)
,Y_i^{(1)} \equiv \log \left(1-p+p \mathrm{e}^{|\zeta| - |\zeta - \mu|}\right),
$ and 
$
    X_i^{(2)} \equiv -\log (1-p+p \mathrm{e}^{|\zeta| - |\zeta - \mu|})
$, $Y_i^{(2)} \equiv -\log \left(1-p+p \mathrm{e}^{|\xi| - |\xi - \mu|}\right),
$
where $\xi \sim \text{Lap}(0,1), \zeta \sim p\text{Lap}(\mu, 1) + (1-p)\text{Lap}(0,1)$.
Note that all the PLLRs are bounded and thus sub-Gaussian. This implies that we can apply Lemma \ref{lem:first_order_edgeworth_bound} directly and also bound the tail similar to Theorem \ref{thm:sub_gaussian_bound}.

\begin{proposition}
Let $\eta  = -\max\{\E(X_i^{(1)}), \E(X_i^{(2)})\} > 0$. 
The tail of the sum of both sequences of PLLRs under the Laplace Mechanism has the following inverse exponential behavior, 
$    \max\left\{ \mathbb{P}\left(\sum_{i = 1}^m X_i^{(1)} \ge \epsilon
    \right), \mathbb{P}\left(\sum_{i = 1}^m X_i^{(2)} \ge \epsilon
    \right) \right\}\le \exp\left(-\frac{2(\epsilon + m \eta)^2}{m \tau^2}
    \right),
$
where $\tau^2 = (\log(1-p+pe^\mu)-\log(1-p+pe^{-\mu}))^2.$
\end{proposition}

\section{Numerical Experiments}
\label{sec:experiments}
In this section, we illustrate the advantages of Edgeworth Accountant by presenting the plots of DP accountant curves under different settings. Specifically, we plot the privacy curve of $\epsilon$ against the number of algorithms under composition and compare our methods (AEA and EEAI) with existing DP accountants. We provide the implementation of our Edgeworth Accountant in \Cref{app:algo_EA}.

\begin{figure*}[h!]
    \centering
    \includegraphics[width = 0.325\linewidth]{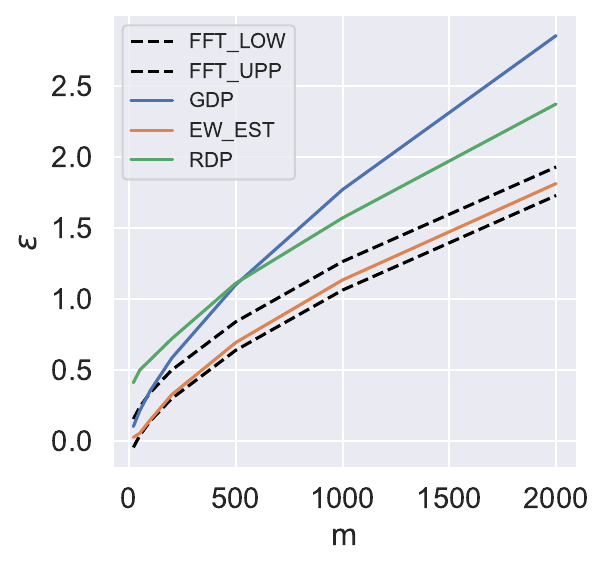}
    \includegraphics[width = 0.325\linewidth]{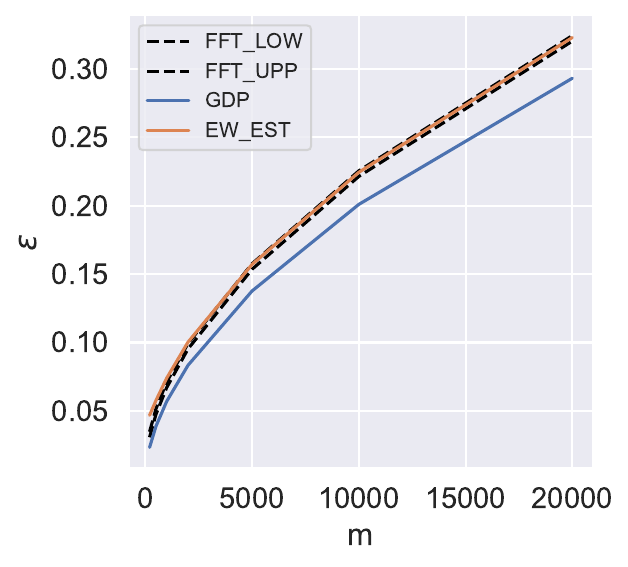}
    \includegraphics[width = 0.325\linewidth]{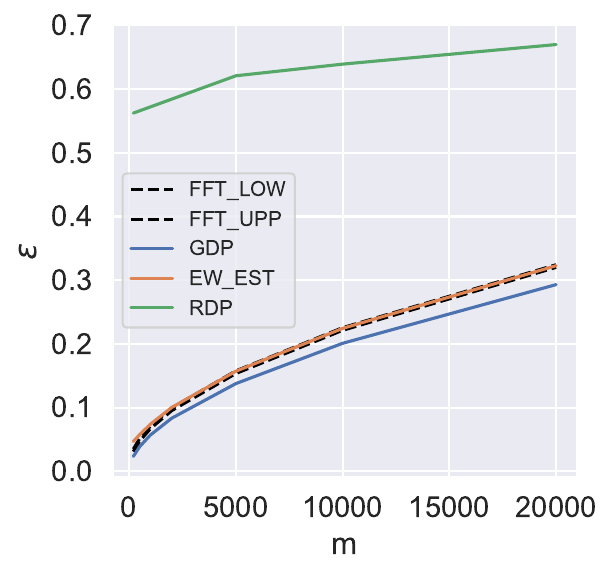}
    \vspace{-0.6em}
    \caption{
    The privacy curve computed via several different accountants. Left: The setting in Figure 1(b) in~\cite{gopi2021numerical}, where $p = 0.01$, $\sigma = 0.8$, and $\delta = 0.015$. Middle and Right: The setting of a real application task in federated learning for $10$ epochs, with $p = 0.05$, $\sigma = 1$, and $\delta = 10^{-5}$. Here, ``EW\_EST'' is the estimate given by our approximate Edgeworth accountant, which provides asymptotically accurate estimates in contrast to exact finite-sample methods such as FFT-based approaches. The RDP curve is omitted from the right panel to facilitate clearer comparison among remaining methods.}
    \vspace{-1em}\label{fig:comparion_to_DGP_RDP}
\end{figure*}

\begin{figure*}[h]
    \centering
    \includegraphics[width = 0.325\linewidth]{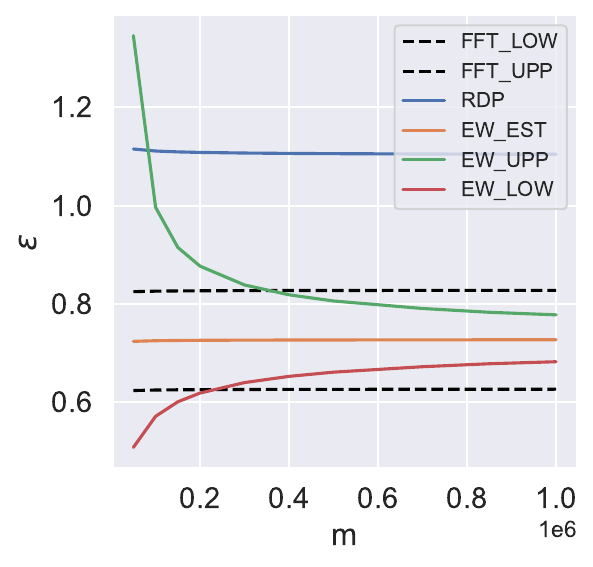}
    \includegraphics[width = 0.325\linewidth]{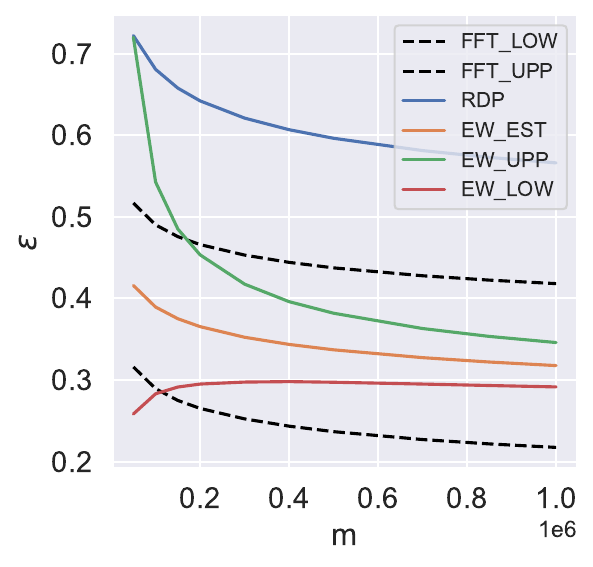}
    \includegraphics[width = 0.325\linewidth]{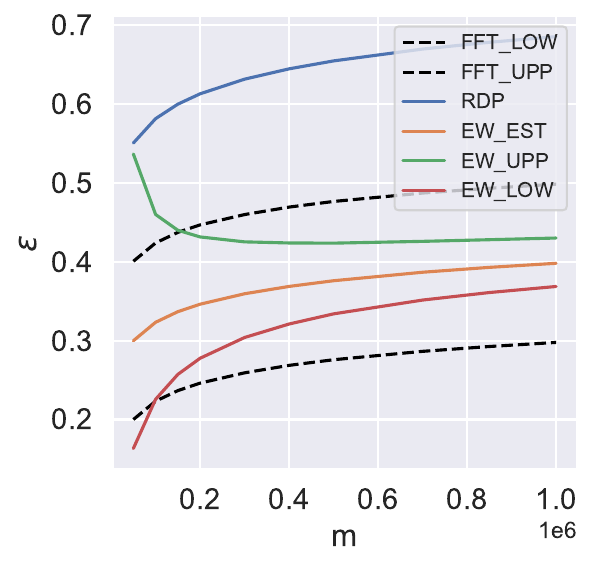}
    \vspace{-0.6em}
    \caption{
    We demonstrate the comparisons between our Edgeworth accountant (both AEA and EEAI), the RDP accountant, and the FFT accountant (whose precision of $\epsilon$ is set to be $0.1$). The three settings are set so that the privacy guarantees does not change dramatically as $m$ increases. Specifically, in all three settings, we set $\delta = 0.1$, $\sigma = 0.8$, and $p = 0.4 / \sqrt{m}$ (left), $p = 1/\sqrt{m\log m}$ (middle), and $p = 0.1\sqrt{\log m/m}$ (right). 
    We omit the GDP curve, because the performance is close to the AEA (``EW\_EST'' curve) when $m$ is large. 
    }
    \label{fig:EEAIs}
\end{figure*}

\noindent\textbf{The AEA.}
We first demonstrate that our proposed approximate Edgeworth Accountant (AEA) is indeed very accurate, outperforming the existing R\'enyi DP and the CLT approximations in experiments. The first experiment has the same setting as in Figure 1(b) in~\cite{gopi2021numerical}, where the authors report that both RDP and GDP are inaccurate, whereas the second setting corresponds to a real federated learning task.
The results are shown in Figure \ref{fig:comparion_to_DGP_RDP}, where we describe the specific settings in the caption. For each sub-figure, the dotted lines ``FFT\_LOW'' and ``FFT\_UPP'' denote the lower and upper bound computed by FFT~\citep{gopi2021numerical} which are used as the underlying ground truth. The ``GDP'' curve is computed by the CLT approximation~\citep{bu2019deep}, the ``RDP'' curve is computed by moments accountant using R\'enyi DP with subsampling amplification~\citep{wang2019subsampled}, and the ``EW\_EST'' curve is computed by our (second-order) AEA. As is evident from the figures, our AEA outperforms GDP and RDP in terms of accuracy. 

\noindent\textbf{The EEAI.} We now present the empirical performance of the EEAI obtained in Section \ref{sec:exact_Edgeworth_Accountant}. We still experiment with NoisySGD. Details of the experiments are given in the caption of Figure \ref{fig:EEAIs}. 
The two error bounds of EEAI are represented by ``EW\_UPP'' and ``EW\_LOW'', and all other curves are defined the same as in the previous setting. 



\noindent\textbf{Numerical stability.} Beyond its superior time complexity, the Edgeworth Accountant demonstrates enhanced numerical stability compared to the FFT method based on numerics, particularly when the number of mechanisms under composition is large.\footnote{In all experiments, we utilized version 0.2.0 of the package from \cite{gopi2021numerical}.} The FFT method may encounter numerical issues in heterogeneous composition with a large value of $m$. These issues could either prevent the code from running or, if the code executes successfully, might result in incorrect privacy guarantees. 

An example of this is shown in Figure \ref{fig:numerical_instable}, where the FFT method, despite running successfully, produces inaccurate results. We examine two subsampled Gaussian mechanisms: $M_1$ with a noise multiplier $\sigma = 0.8$ and subsample rate $p_1$, and $M_2$ with a noise multiplier $\sigma = 0.8$ and subsample rate $p_2$. Here, we consider composing $m_1$ instances of $M_1$ with $m_2$ instances of $M_2$. Following the settings in Figure \ref{fig:EEAIs}, we maintain the same precision level of $\epsilon$ for the FFT method at 0.1 (noting that reducing the precision further exacerbates numerical issues). The privacy guarantees computed by both our method and the FFT method are depicted in Figure \ref{fig:numerical_instable}. The FFT method, requiring discrete convolution for each mechanism, becomes numerically unstable as $m$ increases. In contrast, our analytical finite-sample Edgeworth bounds maintain both accuracy and stability, demonstrating the qualitative superiority of Edgeworth-based methods in the large-$m$ regime. These findings yield important practical implications. For small values of $m$, convolution-based accountants such as the FFT method deliver tight and reliable privacy parameter bounds. However, as $m$ increases, the Edgeworth accountant family offers distinct advantages. The AEA provides highly accurate estimates for moderately large $m$ (several hundred compositions), as demonstrated in Figure \ref{fig:comparion_to_DGP_RDP}. For even larger values of $m$, the EEAI serves as an exact, accurate, and numerically stable accountant that maintains reliability across all composition scales.


\begin{figure}
\centering
\includegraphics[width=0.48\linewidth]{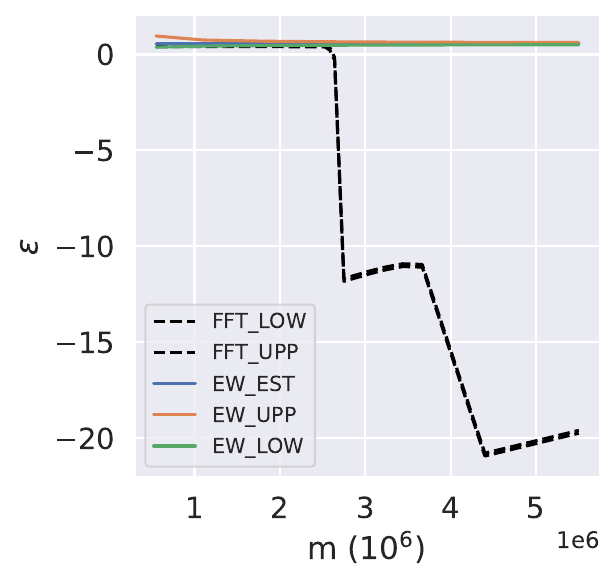}
\includegraphics[width=0.48\linewidth]{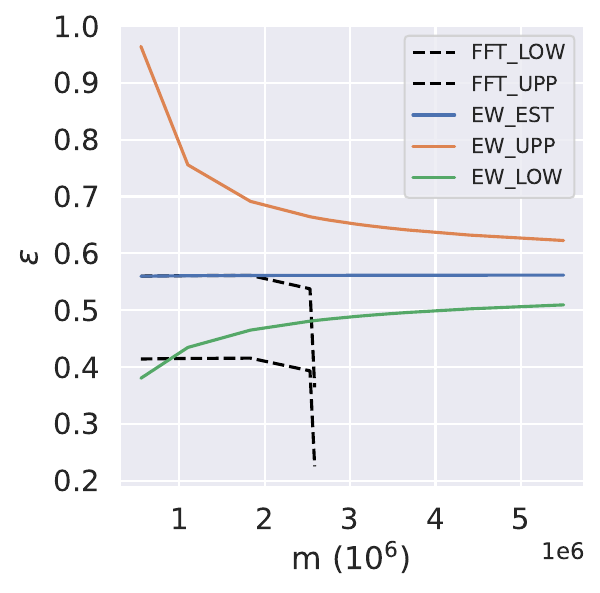}
\caption{With $\delta = 0.1$, $\sigma = 0.8$, $p_1 = 0.35/\sqrt{m_1}$, $p_2 = 0.02/\sqrt{m_2}$, $m_2 = 10m_1$, and $m = m_1 + m_2$, the FFT method exhibits numerical instability and fails to accurately compute $\epsilon$ for large $m$. The left subplot shows that as $m$ increases, the FFT method yields negative $\epsilon$ bounds. The right subplot, truncating all epsilon bounds to positive values, demonstrates that our EEAI provides stable bounds even for such large values of $m$. This numerical instability may stem from the FFT's polynomial dependence on $m$, making numerical stability a significant concern as $m$ becomes very large. In these cases, our EEAI remains numerically stable and effective, even when $m$ is very large.}
\label{fig:numerical_instable}
\end{figure}

\section{Discussion}
\label{sec:discussion}

In this paper, we have developed an analytical approach for accounting privacy loss under the composition of mechanisms, which we name the Edgeworth Accountant. Our approach builds upon the $f$-DP framework, enabling a lossless tracking of privacy guarantees under both subsampling and composition. It is further grounded in an analysis of the probability distributions of PLLRs using Edgeworth expansions. Compared to existing privacy accounting algorithms, our method stands out for its computational efficiency, due to its analytical foundation, and for its numerical stability. We have demonstrated its superior numerical performance in applications related to private deep learning and federated analytics.

As we conclude this paper, we present several open questions for future research. Our experiments indicate that the approximate privacy loss estimated using the AEA is remarkably accurate, while the confidence bands on the true privacy loss offered by the EEAI appear to be conservative. Consequently, an intriguing research direction is the precise quantification of the error in the AEA estimate and the construction of tighter confidence bounds for privacy guarantees. From a practical standpoint, we observe that practitioners often need to run private deep learning models under various hyperparameter configurations and ultimately select one. In such scenarios, swiftly estimating the privacy guarantee for each experimental setup becomes crucial, and the FFT method can be reserved for the final chosen model. This presents a practical use case for our AEA estimator in conjunction with other accounting methods, both numerical and analytical \citep{zhu2021optimal}. An interesting research direction is to develop a principled approach that combines the strengths of our Edgeworth Accountant with other accounting algorithms. 

Furthermore, while our explicit finite-sample bounds are established for first-order expansions, extending these rigorous non-asymptotic guarantees to higher-order expansions is a significant direction. Finally, we emphasize that many algorithms and procedures inherently involve sequential decision-making and naturally lead to extensive mechanism composition. These problems present compelling opportunities for extending the Edgeworth Accountant's applicability beyond NoisySGD and federated analytics to broader classes of privacy-preserving algorithms.

\section*{Acknowledgments}

We are grateful to the anonymous associate editor and referees for their constructive comments, which improved the presentation of this paper. This work was supported in part by a Meta Faculty Research Award, NSF grant DMS-2310679, and Wharton AI for Business.

\bibliographystyle{plain}
\bibliography{ref}


\newpage
\appendix
\onecolumn
\section*{Appendix}

\section{Analysis of NoisySGD}\label{app:algo}
We present the algorithms we considered in \Cref{sec:motivating_applications}. To start with, suppose we have a neural network $h$ that is governed by weights $\w$, with samples $\x_i$ and labels $y_i$ ($i = 1,...,n$). The prediction for each example is $h(\x_i,\w)$, and the per-sample loss is given by $\ell(h(\x_i, \w),y_i)$ for some loss function $\ell$. We define the objective function $L$ to be the average of per-sample losses
\begin{align*}
    L(\w)=\frac{1}{n}\sum_{i=1}^n \ell(h(\x_i, \w),y_i).
\end{align*}
Stochastic Gradient Descent (SGD) algorithm uses a mini-batch as a proxy to this objective function. To control the potential privacy leak in each step of SGD, we need to clip the gradients to control the sensitivity, after which a Gaussian noise is added to it. The details of the algorithm is shown below. 

\begin{center}
\resizebox{0.99\textwidth}{!}{
\begin{minipage}{\textwidth}
\begin{algorithm}[H]
	\caption{NoisySGD (with local or global flat per-sample clipping)}\label{alg:dpsgd}

\begin{algorithmic}[0]
\STATE \textbf{Parameters:} initial weights $\w_0$, learning rate $\eta_t$, subsampling probability $p$, number of iterations $m$, noise scale $\sigma$, gradient norm bound $R$.
		\FOR{$t = 0, \ldots, m-1$}
		\STATE {\hspace{0.55cm}Subsample a batch $I_t\subseteq \{1, \ldots, n\}$ from training set with probability $p$}
		\FOR{$i\in I_t$}
		\STATE \hspace{1.1cm} {$v_t^{(i)} \gets \nabla_{\w} \ell(f(\x_i, \w_t), y_i)$}		
        \STATE \hspace{1.1cm} {$\bar{v}_t^{(i)} \gets \min\big\{1, R/\|v_t^{(i)}\|_2\big\} \cdot v_t^{(i)} $}
        \hfill Clip the gradient
		\ENDFOR
		\STATE \hspace{0.55cm} {$\bar{V}_t\gets\sum_{i\in I_t}\bar{v}_t^{(i)}$}
		\hfill Sum over batch
		\STATE \hspace{0.55cm} {$\w_{t+1} \gets \w_{t} - \frac{\eta_t}{|I_t|} \left(\bar{V}_t+\sigma R \cdot \mathcal{N}(0, I)\right)$}
		\hfill Apply Gaussian mechanism and descend
		\ENDFOR
	\STATE {\bf Output} $\w_{m}$
	\end{algorithmic}
\end{algorithm}
\end{minipage}
}
\end{center}

Recall that in \Cref{sec:eps_bound}, in order to transfer the bounds from CDF approximations to privacy parameters, we need to find a range that contains all possible roots of $\delta = g^+(\epsilon), \delta = g^-(\epsilon)$. Here we showcase how to find such bound in the case of NoisySGD. 
\begin{remark}\label{rmk:bound_on_g}
For NoisySGD, we can express such range analytically. Specifically, for any $\alpha\in \{1, 2\}$ (the index of the sequence of PLLRs), we focus on finding roots in the range $[0, C]$ for $\epsilon^{(\alpha)+}$ and $\epsilon^{(\alpha)-}$, where $C$ is the smallest value of $\epsilon$ such that 
\begin{equation*}
    C \ge \sup_{S\subset \mathbb{R}, \P(Y^{(\alpha)}\in S) \ge \delta} \log\left(\frac{\P(Y^{(\alpha)}\in S)}{\P(X^{(\alpha)}\in S)}\right).
\end{equation*}
This is clearly a (sub-optimal) upper bound. A loose bound of the range can be easily proved to be
\begin{equation*}
C = \min\left\{m \log\left(p \frac{\delta}{1-\Phi(z_\delta + \mu)}\right), \log\left(\frac{\delta}{1-\Phi\left(\frac{z_{{\delta}/{\sqrt{m}}} + \mu}{\sqrt{m}}\right)}\right)\right\},
\end{equation*}
where $z_\delta$ is the upper $\delta$ quantile of a standard normal distribution. We can find $0<\epsilon^{(\alpha)-} \le \epsilon^{(\alpha)+}$ in the range defined above.

\end{remark}

\section{Error Rate Analysis for Numerical Integration}\label{app:error_analysis}

In Table \ref{table:table1}, the Edgeworth Accountant has a computational complexity of $O(1)$ for homogeneous compositions and $O(m)$ for heterogeneous compositions in mathematical analysis. However, in practice, this only holds up to a logarithmic factor due to numerical errors from modern computer discretization. For example, storing the number $m$ alone requires $O(\log m)$ space and time complexity; and the simple addition of $m$ real numbers can incur an order $O(m)$ numerical error. To account for these numerical errors, the achievable time complexity for our homogeneous composition is $O(\log m) = \widetilde{O}(1)$, and heterogeneous composition is $O(m\log m) = \widetilde{O}(m)$, where $\widetilde{O}$ omits logarithmic terms of $m$. 

We now demonstrate the assumptions we need and how to calculate the time complexity when considering numerical errors for EEAI (Algorithm \ref{alg:EEAI}). The analysis for AEA is similar.
In practice, to implement Algorithm \ref{alg:EEAI}, we need to know $f_i$ for each algorithm $M_i$ as an oracle, this is a very mild assumption when we know the noise addition mechanism used in each composing algorithm. We further assume the evaluation of any finite moments of PLLRs corresponding to the $f_i$ is in constant time. This is the same model we used for comparing all the methods, including FFT. Specifically, this is almost obvious with an oracle of $f_i$: we can obtain the analytical form of the PLLRs as specified in Definition \ref{def:privacy_loss_loglikelihood}, and we approximate the DP guarantee with the corresponding closed-form expression as specified in the Algorithm \ref{alg:EEAI}. 
The only term that may potentially depend on $m$ here is the accumulated numerical error. With the order of $k = 1$, we only need to numerically calculate up to 4th moment of the PLLRs for a fixed number of $\alpha$'s (for subsampled mechanism, $|\mathcal{I}| = 2$, and non-subsampled mechanism, $|\mathcal{I}| = 1$). We then sum these numerically integrated results together to calculate $G_{m, k, X}$ and $\Delta_{m, k, X}$. When $m = 1$, in order to achieve a discretization error of $\tau$, there is a minimum number of partitions in the discretization required. Let this be $n_1$.  When $m > 1$, the worst case discretization error of the summation of the moment becomes $O(m \tau)$.  To achieve the same discretization error for $m > 1$ and $m= 1$, each integration needs to have a discretization error of $\tau’ = \frac{\tau}{m}$, which requires the partition to be of order $o(n_1 \log m)$. To see this, we can do integration using any standard  method like Gaussian quadrature: with $n$ partition in the discretization, the error is bounded by a rate of $O(n^{-n})$ \footnote{Here we start from the classic error estimate of the integral of $f$ on $[a, b]$ of form $\frac{(b - a)^{2n + 1}(n!)^4}{(2n + 1) (2n!)^3}f^{(2n)}(\eta)$, for some $\eta\in[a, b]$. We need the density of the PLLR to be smooth enough, i.e. the higher order derivative is not growing faster than $c^n$ for some constant $c > 1$. This condition on $f$ can be verified for common mechanisms like Gaussian or Laplacian.}. Therefore, to achieve the same discretization error for $m > 1$ and $m= 1$, we need an $o(\log m)$ multiplicative overhead to our computation time of each numerical integration. However, the time to express $m$ is even larger than this overhead. Thus, we need an $O(\log m)$ term for both heterogeneous and homogeneous compositions, resulting in a time complexity of $O(\log m)$ and $O(m\log m)$ for homogeneous and heterogeneous compositions, respectively.

\section{Proofs in Section \ref{sec:PLLRs}}
\label{app:proof_sec_2}
\subsection{Proof of Proposition \ref{prop:eps_delta_f_dual}}\label{app:proof_sec_C_1}
We present the proof of Proposition \ref{prop:eps_delta_f_dual} in this section. The proof relies on two Lemmas that are of self-interest and we first present the lemmas. The proof of Proposition \ref{prop:eps_delta_f_dual} is straightforward from results of Lemmas. 
Recall that the trade-off functions $f_i = T(P_i, Q_i)$ we consider are realized by the two following hypotheses:
\begin{equation*}
H_{0, i}: w_i\sim P_i \text { vs. } H_{1, i}: w_i\sim Q_i, 
\end{equation*}
where $P_i, Q_i$ are two distributions. To evaluate the trade-off function $f = \bigotimes_{i = 1}^m f_i$, we are essentially distinguishing between the two composite hypotheses
\begin{equation}
    H_0: \bm{w} \sim P_1\times P_2 \times \cdots \times P_m ~\text{ vs. }~ H_1: \bm{w} \sim Q_1\times Q_2 \times \cdots \times Q_m,
\end{equation}
where $\bm{w} = (w_1, ..., w_m)$ is the concatenation of all the $w_i$'s. 
The following lemma shows how to connect PLLRs of each $f_i$ to the trade-off function $f$.
\begin{lemma}
\label{lem:del_eps_naive_parametric_form}
Let $X_{1}, \ldots, X_{m}$ be the PLLR under the null hypothesis and, likewise, $Y_{1}, \ldots, Y_{m}$ be the PLLR under the alternative. Let $F_{X, m}, F_{Y, m}$ be the CDFs of $x\equiv X_{1}+\ldots+X_{m}$ and $Y\equiv Y_{1}+\cdots+Y_{m}$, respectively. Then we have the following relationship between privacy parameters and privacy-loss log-likelihood ratios
\begin{equation}
\begin{aligned}
\label{eq:eps_delta_naive_relation}
\epsilon &=\log \frac{F_{Y, m}^{\prime}(c)}{F_{X, m}^{\prime}(c)}, \\
\delta &=\frac{F_{X, m}^{\prime}(c)\left(1-F_{Y, m}(c)\right)-F_{Y, m}^{\prime}(c)\left(1-F_{X, m}(c)\right)}{F_{X, m}^{\prime}(c)},
\end{aligned}
\end{equation}
where $c$ is some constant.
\end{lemma}

\begin{proof}[Proof of Lemma \ref{lem:del_eps_naive_parametric_form}]
To distinguish between $H_0: P_1\times P_2 \times ... \times P_m$ vs.  $H_1: Q_1\times Q_2 \times ... \times Q_m$, By the Neyman-Pearson lemma, we know that each point of the trade-off function $f$ is realized by a likelihood ratio test (cut-off at some threshold $c$). So, the trade-off function takes a parametric form $f(\alpha)=\beta$, 
where $\alpha$ is the type-I error of the test, and $\beta$ is type-II error of the test:
\begin{align*}
\alpha&=\mathbb{P}_{H_0}\left(\log\left(\frac{d P_1\times P_2 \times \cdots \times P_m}{d Q_1\times Q_2 \times \cdots \times Q_m}(\bm{w})\right) > c\right) \\
\beta&=\mathbb{P}_{H_1}\left(\log\left(\frac{d P_1\times P_2 \times \cdots \times P_m}{d Q_1\times Q_2 \times \cdots \times Q_m}(\bm{w})\right) \le c\right)
\end{align*}
Note that under $H_0$, we have
\begin{align*}
    \log\left(\frac{d P_1\times P_2 \times \cdots \times P_m}{d Q_1\times Q_2 \times \cdots \times Q_m}(\bm{w})\right) &= \log\left(\frac{d P_1}{dQ_1}(w_1) \times\cdots\times\frac{d P_m}{dQ_m}(w_m)\right)\\
    &= \log\left(\frac{d P_1}{dQ_1}(w_1)\right) +  \cdots + \log\left(\frac{d P_m}{dQ_m}(w_m)\right)\\
    &= X_1 + \cdots + X_m = X.
\end{align*}
and similarly under $H_1$, 
\begin{align*}
    \log\left(\frac{d P_1\times P_2 \times \cdots \times P_m}{d Q_1\times Q_2 \times \cdots \times Q_m}(\bm{w})\right) 
    &= Y_1 + \cdots + Y_m = Y.
\end{align*}
So, we can simplfy the parametric form of $f$ by $f(\alpha) = \beta$, where
\begin{align*}
\alpha&=\mathbb{P}\left(X_{1}+\cdots+X_{m}>c\right) =1-F_{X, m}(c) \\
\beta&=\mathbb{P}\left(Y_{1}+\cdots+Y_{m} \leq c\right) = F_{Y, m}(c).
\end{align*}

This allows us to simply write
\begin{equation*}
f(\alpha)=F_{Y, m} \circ F_{X, m}^{-1}(1-\alpha).
\end{equation*}
For a point $(\alpha, \beta)$ on the trade-off function $f$, where
\begin{equation*}
\beta=F_{Y, m}(c)=F_{Y, m}\left(F_{X, m}^{-1}(1-\alpha)\right),
\end{equation*}
and $\alpha$ is small. By the equivalence given in Proposition 2.5 in~\cite{dong2019gaussian}, we know that the slope of the tangent line passing through $(\alpha, \beta)$ (for small $\alpha$) is given by
\begin{equation*}
\begin{aligned}
-\mathrm{e}^{\epsilon} = \frac{df}{d\alpha}(\alpha)&=F_{Y, m}^{\prime}\left(F_{X, m}^{-1}(1-\alpha)\right) \cdot \frac{1}{F_{X, m}^{\prime}\left(F_{X, m}^{-1}(1-\alpha)\right)} \cdot(-1) =-\frac{F_{Y, m}^{\prime}\left(F_{X, m}^{-1}(1-\alpha)\right)}{F_{X, m}^{\prime}\left(F_{X, m}^{-1}(1-\alpha)\right)},
\end{aligned}
\end{equation*}
which gives
\begin{equation*}
\epsilon=\log \frac{F_{Y, m}^{\prime}\left(F_{X, m}^{-1}(1-\alpha)\right)}{F_{X, m}^{\prime}\left(F_{X, m}^{-1}(1-\alpha)\right)}=\log \frac{F_{Y, m}^{\prime}(c)}{F_{X, m}^{\prime}(c)}.
\end{equation*}
The equation of the tangent line takes the form of
\begin{equation*}
y=-\frac{F_{Y, m}^{\prime}(c)}{F_{X, m}^{\prime}(c)}\left(x-1+F_{X, m}(c)\right)+F_{Y, m}(c).
\end{equation*}
Using the same proposition,we know the intercept of the line is $1-\delta$, so we should have
\begin{equation*}
1-\delta=\frac{F_{Y, m}^{\prime}(c)\left(1-F_{X, m}(c)\right)}{F_{X, m}^{\prime}(c)}+F_{Y, m}(c)=\frac{F_{Y, m}^{\prime}(c)\left(1-F_{X, m}(c)\right)+F_{X, m}^{\prime}(c) F_{Y, m}(c)}{F_{X, m}^{\prime}(c)},
\end{equation*}
which gives
\begin{equation*}
\begin{aligned}
\delta &=1-\frac{F_{Y, m}^{\prime}(c)\left(1-F_{X, m}(c)\right)+F_{X, m}^{\prime}(c) F_{Y, m}(c)}{F_{X, m}^{\prime}(c)} \\
&=\frac{F_{X, m}^{\prime}(c)\left(1-F_{Y, m}(c)\right)-F_{Y, m}^{\prime}(c)\left(1-F_{X, m}(c)\right)}{F_{X, m}^{\prime}(c)}.
\end{aligned}
\end{equation*}
Therefore, $\epsilon$ and $\delta$ takes the following parametric form as in the statement of the lemma,
\begin{equation*}
\begin{aligned}
\epsilon &=\log \frac{F_{Y, m}^{\prime}(c)}{F_{X, m}^{\prime}(c)} \\
\delta &=\frac{F_{X, m}^{\prime}(c)\left(1-F_{Y, m}(c)\right)-F_{Y, m}^{\prime}(c)\left(1-F_{X, m}(c)\right)}{F_{X, m}^{\prime}(c)}.
\end{aligned}
\end{equation*}
\end{proof}

To simplify the relation in \eqref{eq:eps_delta_naive_relation}, we observe the following interesting lemma about PLLRs.

\begin{lemma}
\label{lem:eps_delta_simplify}
 Let $X_1, X_2, \dots, X_m$ and $Y_1, Y_2, \dots, Y_m$ and $F_{X,m}, F_{Y,m}$ be defined as in Lemma \ref{lem:del_eps_naive_parametric_form}. Let $f_{X,m}, f_{Y,m}$ be the PDFs of $\sum_{i=1}^m X_i$ and $\sum_{i=1}^m Y_i$. Then we have for any $c\in \mathbb{R}$,
 \begin{equation}\label{eq:fy_fx_relation}
    c=\log\frac{f_{Y,m}(c)}{f_{X,m}(c)}.
\end{equation}
\end{lemma}
\begin{proof}[Proof of Lemma \ref{lem:eps_delta_simplify}]
We use induction on $m$, the number of compositions, to prove this Lemma. 
\par\textit{Base Case: $m=1$.} When $m = 1$, we write out the forms of $X$ and $Y$ explicitly as \begin{equation*}
    X = \log\frac{Q_1(w_1)}{P_1(w_1)} \quad\text{where $w_1\sim P_1$},
\end{equation*}
\begin{equation*}
    Y = \log\frac{Q_1(w_1)}{P_1(w_1)} \quad\text{where $w_1\sim Q_1$}.
\end{equation*}
As a result, for any measurable function $g:\mathbb{R}\rightarrow\mathbb{R}$, we have\begin{align*}
    \E_Y[g(Y)]&=\E_{w_1\sim Q_1}\left[g\left(\log\frac{Q_1(w_1)}{P_1(w_1)}\right)\right]\\
    &=\E_{w_1\sim P_1}\left[g\left(\log\frac{Q_1(w_1)}{P_1(w_1)}\right)\frac{Q_1(w_1)}{P_1(w_1)}\right]\\
    &=\E_{w_1\sim P_1}\left[g\left(\log\frac{Q_1(w_1)}{P_1(w_1)}\right)e^X\right]\\
    &=\E_X[g(X)e^X].
\end{align*}
Since the above equality holds for all $g$, we must have that their exists a version of both PDFs such that $f_{Y, 1}(t) = f_{X, 1}(t)e^t$. This shows that for $m=1$,\begin{equation*}
     c=\log\frac{f_{Y, 1}(c)}{f_{X, 1}(c)}.
\end{equation*}
\par\textit{Induction Step: }Suppose the result is true for $m$, we now show that it is also true for $m+1$ compositions. We now claim the following Lemma.
\begin{lemma}
\label{lemma:con_prv}
Let $A_1, A_2, B_1, B_2$ be four random variables.
Denote the PDFs of $A_1, A_2, B_1, B_2$ by $f_{A_1}, f_{A_2}, f_{B_1}, f_{B_2}$ respectively. Suppose further that \begin{equation*}
    f_{B_1}(t)=g(t)f_{A_1}(t), \quad f_{B_2}(t)=g(t)f_{A_2}(t)\quad\text{for all $t$,}
\end{equation*}
for some function $g$ satisfying $g(x+y)=g(x)g(y)$. Let $f_{A,2}, f_{B,2}$ denote the density function for $A_1+A_2, B_1+B_2$. Then \begin{equation*}
     f_{B,2}(t)=g(t)f_{A,2}(t)\quad\text{for all $t$}.
\end{equation*}
\end{lemma}
Proof of Lemma \ref{lemma:con_prv} will be given at the end of the proof. Applying Lemma \ref{lemma:con_prv} on random variables $A_1 = \sum_{i=1}^m X_i, A_2 = X_{m+1}$ and $B_1= \sum_{i=1}^m Y_i, B_2 = Y_{m+1}$, we will show that we get the desired relationship for $m+1$ compositions. By induction hypothesis we know that $f_{B_1}(t) = g(t)f_{A_1}(t)$ for $g(t) = e^t$. Since $f_{Y_{m+1}}(t) = g(t)f_{X_{m+1}}(t)$ by the base case in induction, and $g(x+y)=g(x)g(y)$, we have $f_{Y,m+1}(t)=g(t)f_{X,m+1}(t)$ for all $t$. This indicates that we have\begin{equation*}
    c = \log \frac{f_{Y,m+1}(c)}{f_{X,m+1}(c)}
\end{equation*}
for any $c$. Hence we have completed the induction step and concluded the proof. 

\end{proof}

\begin{proof}[Proof of Lemma \ref{lemma:con_prv}]
We use the convolution formula on $B_1, B_2$ and obtain\begin{align*}
    f_{B,2}(t) &= \int_{-\infty}^\infty f_{B_1}(t-u)f_{B_2}(u)du\\
    &=\int_{-\infty}^\infty g(t-u)g(u)f_{A_1}(t-u)f_{A_2}(u)du\\
    &=\int_{-\infty}^\infty g(t)f_{A_1}(t-u)f_{A_2}(u)du\\
    &=g(t)f_{A,2}(t)
\end{align*}
by convolution formula on $A_1, A_2$. 

\end{proof}

\begin{proof}[Proof of Proposition \ref{prop:eps_delta_f_dual}]
By Lemma \ref{lem:eps_delta_simplify}, $\frac{F'_{Y,m}(c)}{F'_{X,m}(c)}=e^c$, so the parameter $c$ in Lemma  \ref{lem:del_eps_naive_parametric_form} equals $\epsilon$. Substituting $c=\epsilon$ back to (\ref{eq:eps_delta_naive_relation}) gives
\begin{equation}
    \delta(\epsilon) = 1-F_{Y,m}(\epsilon)-e^{\epsilon}(1-F_{X,m}(\epsilon)).
\end{equation}
This is exactly the $\delta_{X,Y}(\epsilon)$ we defined, and it fully characterizes the trade-off function associated with the ordered pair $(P_1\times P_2 \times \cdots \times P_m,Q_1\times Q_2 \times \cdots \times Q_m)$. 

In particular, when the trade-off function is symmetric, the final equivalence on $\epsilon\ge 0$ follows from Fact \ref{fact:1}.

\end{proof}

\begin{proof}[Proof of Lemma \ref{lem:composite_duality}.]
Define \begin{equation*}
    h(x) = \left(\inf_{\alpha\in\mathcal{I}} f^{{(\alpha)}}\right)^*(x),
\end{equation*}
which is convex and lower semi-continuous by definition of convex conjugate. By Fenchel–Moreau theorem, we have $h^{**} = h$. Denote $(\epsilon, \delta(\epsilon))$ to be the equivalent dual relationship to $f = \inf_\alpha\{f^{(\alpha)}\}^{**}$-DP. Using the convex dual characterization of trade-off functions as in  \cite{dong2019gaussian}, we have \begin{equation*}
    \delta(\epsilon) = 1 + \left(\inf_{\alpha\in\mathcal{I}} f^{{(\alpha)}}\right)^{***} (-e^{\epsilon}) =1+h^{**}(-e^{\epsilon})= 1+ h(-e^{\epsilon}).
\end{equation*}
By order reversing property of convex conjugate, we have\begin{align*}
    \delta(\epsilon)&= 1+h(-e^{\epsilon})\\
    &= 1 +\left(\inf_{\alpha\in\mathcal{I}} f^{(\alpha)}\right)^*(-e^{\epsilon})\\
    &= 1 + \sup _{\alpha\in\mathcal{I}} f^{(\alpha)*}(-e^{\epsilon})\\
    &= \sup _{\alpha\in\mathcal{I}}  \left(1 + f^{(\alpha)*}(-e^{\epsilon})\right)\\
    &= \sup _{\alpha\in\mathcal{I}} \delta^{(\alpha)}(\epsilon)
\end{align*}
where we used dual relationship for each $\alpha\in\mathcal{I}$ again in the last step. And the other direction follows directly from the duality of $f$-DP and $(\epsilon, \delta(\epsilon))$-DP, meaning that if the mechanism satisfies $(\epsilon, \sup_{\alpha\in \mathcal{I}}\delta^{(\alpha)})$-DP, then it also satisfies $f = \left(\inf_\alpha\{f^{(\alpha)}\}\right)^{**}$-DP.

In particular, when $f$ is symmetric, the equivalence with $f$-DP on $\epsilon\ge 0$ follows from Fact \ref{fact:1}.
\end{proof}

\section{Proofs in Section \ref{sec:edgeworth_approximation_bound}}
\label{app:proof_sec_4}
\subsection{Proof of Theorem \ref{thm:sub_gaussian_bound}}

\begin{proof}[Proof of Theorem \ref{thm:sub_gaussian_bound}]
We first briefly introduce the idea of the proof. The main idea is to construct a random variable $\tX_i$ by choosing an $a\ge 0$, such that it stochastically dominates $X_i$ (that is, $X_i \le \tX_i$ a.s.), and satisfies $\mathbb{E}(\tX_i) < 0$. We then choose $\eta(a) = -\mathbb{E}(\tX_i)$ which is a positive number. In what follows, we will explicitly construct $\tX_i$ so that $\tX_i$ can be decomposed into the sum of two sub-Gaussian random variables with parameters $\sigma_A^2, \sigma_B^2$.
Then since $X_i\leq \tX_i$ a.s., we deduce that\begin{equation*}
  \mathbb{P}\left(\sum_{i = 1}^m X_i \ge \epsilon\right)\leq \mathbb{P}\left(\sum_{i = 1}^m \tX_i \ge \epsilon\right) = \mathbb{P}\left(\sum_{i = 1}^m \tX_i -\sum_{i = 1}^m \E(\tX_i)\ge \epsilon + m\eta\right),
\end{equation*}
The final conclusion, which will be proved at the end, follows from the sub-Gaussian bounds.

For notation-wise convenience, we first define a quantity depending on the value of $\xi_i$, where \begin{equation*}
    \Delta(\xi_i) := X_i - (\xi_i\mu-\frac{1}{2}\mu^2) =  \log\left(p+\frac{1-p}{e^{\mu\xi_i-\frac{\mu^2}{2}}}\right).
\end{equation*}
It is obvious that $\Delta(\xi_i)$ is a strictly decreasing function of the value of $\xi_i$. Now, we construct the random variable $\tX_i$ as follows. Define
\begin{equation*}
\tX_i = A_i + B_i
\end{equation*}
where
\begin{equation}
    A_i=\begin{cases}
X_i &\text{if $\xi_i < a$,}\\
a^+\mu-\frac{1}{2}\mu^2 + \Delta(a) &\text{if $\xi_i \geq a$}.
\end{cases}\label{eq:Ai}
\end{equation}
and 
\begin{equation}
    B_i=\begin{cases}
0 &\text{if $\xi_i < a$,}\\
{\xi}_i\mu-a^+\mu &\text{if $\xi_i \geq a$}.
\end{cases}\label{eq:Bi}
\end{equation}
Here, we define $a^+ = \frac{\phi(a)}{1 - \Phi(a)}$. Note that $a^+>a$ for any $a>0$ by bounds on Mills ratio. To shed light on this decomposition, we first show that $\tX_i$ stochastically dominates $X_i$. Since $\Delta(\xi_i)$ is a decreasing function in $\xi_i$, hence when $\xi_i\geq a$, we have $\Delta(\xi_i) \leq  \Delta(a).$ As a result, when $\xi_i>a$,\begin{equation*}
    \tX_i = \xi_i\mu-\frac{1}{2}\mu^2 + \Delta(a) \geq  \xi_i\mu-\frac{1}{2}\mu^2 + \Delta(\xi_i)=X_i.
\end{equation*}
This guarantees that $\tX_i\geq X_i$ a.s.. Another good property of this decomposition, we observe that $A_i$ is sub-Gaussian due to Lemma \ref{lem:Ai_SG}, and $B_i$ is a mean-zero sub-Gaussian random variable from Lemma \ref{lem:Bi_SG}. Proof of the two Lemmas is postponed to the next section.

\par Note that the above construction is valid for any $a$. Now we show that there exists some $a>0$ such that $\E(\tX_i) = -\eta(a)<0$ where $\eta(a)$ only depends on $a$. Note that \begin{equation*}
    \E(\tX_i) - \E(X_i) = \int_{a}^\infty (\Delta(a) - \Delta(\xi)) \phi(\xi)d\xi,
\end{equation*}
where $\phi(x)$ is the density for standard Normal random variable. Also recall that $\E(X_i) < 0$. Then \begin{align*}
    \E(\tX_i) &= \E(X_i) +\int_{a}^\infty (\Delta(a) - \Delta(\xi)) \phi(\xi)d\xi\\
    & = \E(X_i) + e(a),
\end{align*}
where $e(a)$ satisfies that $\lim_{a\rightarrow \infty} e(a) = 0.$ This is because by construction $\int_{a}^\infty (\Delta(a) - \Delta(\xi)) \phi(\xi)d\xi \geq 0$ and that \begin{equation*}
    e(a)= \int_{a}^\infty \Delta(a)\phi(\xi)d\xi - \int_{a}^\infty \Delta(\xi)\phi(\xi)d\xi.
\end{equation*}
The second term vanishes as $\xi\rightarrow \infty$ since $\Delta(\xi)$ is integrable. For the first term, if $\Delta(a)<0$ the integral is already negative. If $\Delta(a)>0$ we have $\int_{a}^\infty \Delta(a)\phi(\xi)d\xi< \int_{a}^\infty \Delta(0)\phi(\xi)d\xi$ which also vanishes. As a result, we have shown that $\lim_{a\rightarrow\infty} e(a)\leq 0$. Combined with what we have above, we deduce that $\lim_{a\rightarrow \infty} e(a) = 0$ as required. Then we can pick $a$ large enough such that $e(a) = -\frac{1}{2}\E(X_i)$ and we have $\E(\tX_i) = \frac{1}{2}\E(X_i)<0.$ 
\par Now we can combine the previous results and prove the tail bound of $\sum_{i=1}^m X_i. $ Recall we have constructed random variable $\tX_i = A_i+B_i$ such that $\tX_i\geq X_i$ a.s. with $\E(\tX_i)=-\eta(a)<0$. Moreover, $A_i, B_i$ are both sub-Gaussian with parameters $\sigma_A^2$ and $\sigma_B^2$. Then we have\begin{align*}
    \mathbb{P}\left(\sum_{i = 1}^m X_i \ge \epsilon\right)&\leq \mathbb{P}\left(\sum_{i = 1}^m \tX_i \ge \epsilon\right) = \mathbb{P}\left(\sum_{i = 1}^m \tX_i -\sum_{i = 1}^m \E(\tX_i)\ge \epsilon + m\eta\right)\\
    &\leq \mathbb{P}\left(\sum_{i = 1}^m A_i -\sum_{i = 1}^m \E(A_i) + \sum_{i = 1}^m B_i -\sum_{i = 1}^m \E(B_i)\ge \epsilon + m\eta\right)\\
    &\leq \mathbb{P}\left(\sum_{i = 1}^m A_i -\sum_{i = 1}^m \E(A_i) \ge \frac{\epsilon + m\eta}{2}\right)+ \mathbb{P}\left(\sum_{i = 1}^m B_i -\sum_{i = 1}^m \E(B_i) \ge \frac{\epsilon + m\eta}{2}\right),
\end{align*}
where the last inequality follows from the union bound. Finally, since $A_i, B_i$ are both sub-Gaussian, we know that their sum $\sum_{i = 1}^m A_i, \sum_{i = 1}^m B_i$ are still sub-Gaussian. Hence \begin{equation*}
    \mathbb{P}\left(\sum_{i = 1}^m A_i -\sum_{i = 1}^m \E(A_i) \ge \frac{\epsilon + m\eta}{2}\right)\leq \exp\left(-\frac{(\epsilon + m\eta)^2}{8m\sigma_A^2}\right),
\end{equation*}
\begin{equation*}
    \mathbb{P}\left(\sum_{i = 1}^m B_i -\sum_{i = 1}^m \E(B_i) \ge \frac{\epsilon + m\eta}{2}\right)\leq \exp\left(-\frac{(\epsilon + m\eta)^2}{8m\sigma_B^2}\right).
\end{equation*}
As a result, \begin{equation*}
   \mathbb{P}\left(\sum_{i = 1}^m X_i \ge \epsilon\right)\leq  \exp\left(-\frac{(\epsilon + m\eta)^2}{8m\sigma_A^2}\right)+\exp\left(-\frac{(\epsilon + m\eta)^2}{8m\sigma_B^2}\right)\leq 2\exp\left(-\frac{(\epsilon + m\eta)^2}{8m\tau^2}\right),
\end{equation*}
where \begin{align*}
    \tau^2 &= \max\{\sigma_A^2, \sigma_B^2\} \\
    &= \max\left\{\frac{(\log(1-p+pe^{\mu a-\frac{1}{2}\mu^2})+\mu(a^+-a)-\log(1-p))^2}{4}, \mu^2, \frac{(a^+ - a)^2\mu^2}{2\log(\Phi(a^+) - \Phi(a))}\right\}.
\end{align*}

\end{proof}

\subsection{Technical Lemmas}
\begin{lemma}
\label{lem:Ai_SG}
The random variable $A_i$, defined in \eqref{eq:Ai} is sub-Gaussian random variable with parameter $\sigma_A^2$ where $
    \sigma_A^2 = \dfrac{(\log(1-p+pe^{\mu a-\frac{1}{2}\mu^2})+\mu(a^+-a)-\log(1-p))^2}{4}.$
\end{lemma}

\begin{proof}[Proof of Lemma \ref{lem:Ai_SG}]
The proof of Lemma \ref{lem:Ai_SG} is straightforward, we show that $A_i$ is bounded and thus sub-Gaussian by Hoeffding's inequality. Note that when $\xi_i < a$, we have $A_i = X_i= \log(1-p+pe^{\mu\xi_i-\frac{1}{2}\mu^2}) < \log(1-p+pe^{\mu a-\frac{1}{2}\mu^2})$. Moreover, since $X_i$ is bounded below by $\log(1-p)$, we deduce that when $\xi_i<a$, we have \begin{equation*}
    \log(1-p)<A_i<  \log(1-p+pe^{\mu a-\frac{1}{2}\mu^2}), 
\end{equation*}
which is bounded as desired. 
\par On the other hand, when $\xi_i > a$, by definition of $\Delta(\xi_i)$, \begin{equation*}
    A_i=a^+\mu-\frac{1}{2}\mu^2 + \log\left(p+\frac{1-p}{e^{\mu a-\frac{\mu^2}{2}}}\right)=\log(1-p+pe^{\mu a-\frac{1}{2}\mu^2})+\mu(a^+-a),
\end{equation*}
which is a constant. 
Since $a^+>a$, in this case the above constant is greater than $\log(1-p+pe^{\mu a-\frac{1}{2}\mu^2}).$ Combine the above two settings, we deduce that $A_i\in(\log(1-p), \log(1-p+pe^{\mu a-\frac{1}{2}\mu^2})+\mu(a^+-a))$ is a bounded random variable. By Hoeffding's inequality, it is sub-Gaussian with parameter defined in the Lemma. 

\end{proof}

\begin{lemma}\label{lem:Bi_SG}
The random variable $B_i$ defined in \Cref{eq:Bi} is a mean-zero sub-Gaussian random variable with parameter $\sigma_B^2 = \mu^2 \max\left\{1, \frac{(a^+ - a)^2}{2\log(\Phi(a^+) - \Phi(a))}\right\}$.
\end{lemma}
\begin{remark}
We note that as a function of $a$, $\frac{(a^+ - a)^2}{2\log(\Phi(a^+) - \Phi(a))}$ is in fact a decreasing function, and is always less than $1$ as $a > 0$. Its plot can be found in Figure \ref{fig:Bi_sub-Gaussian_left}. Therefore, by truncating normal at $a$, we essentially loss nothing, since $B_i$ is still a sub-Gaussian random variable with parameter $\mu$.
\end{remark}
\begin{figure}
    \centering
    \includegraphics[width = 0.6\linewidth]{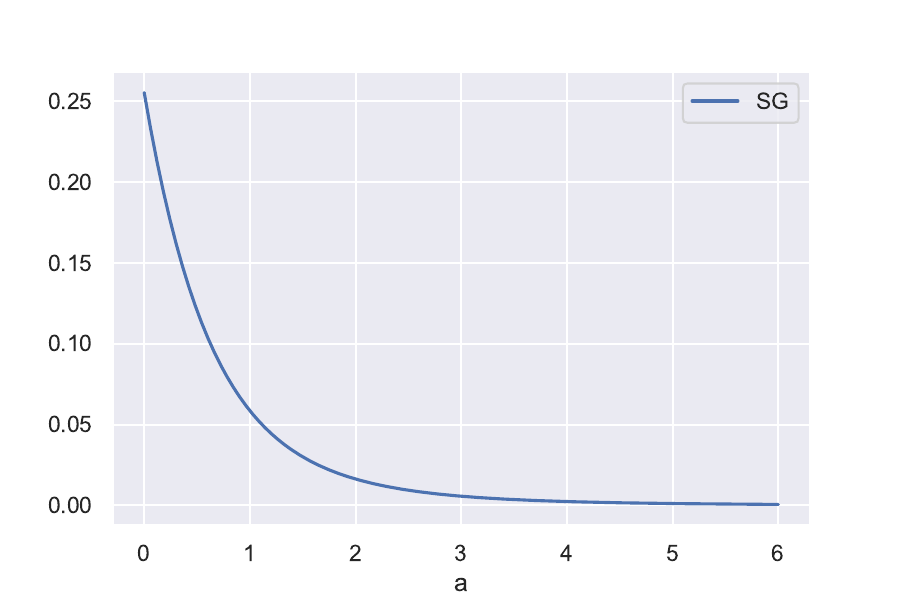}
    \caption{The value of $\frac{(a^+ - a)^2}{2\log(\Phi(a^+) - \Phi(a))}$ when $a > 0$.}
    \label{fig:Bi_sub-Gaussian_left}
\end{figure}
\begin{proof}[Proof of Lemma \ref{lem:Bi_SG}]
Recall the definition in \Cref{eq:Bi}, we can re-write $B_i$ as a mixture random variable 
\begin{equation*}
    {B}_i=\begin{cases}
0 &\text{w.p. $\P(\xi_i < a$),}\\
\widetilde{\xi}_i\mu-a^+\mu &\text{w.p. $\P(\xi_i \geq a)$}.
\end{cases}
\end{equation*}
where $\widetilde{\xi}_i = \xi_i \big|\xi_i >0$ is the normal $\mathcal{N}(0, 1)$ truncated at $a > 0$, whose probability density function is given by
\begin{equation*}
    f(t) = \frac{\phi(t) }{1 - \Phi(a)}, \text{ for }t > a.
\end{equation*}
From Lemma \ref{lem:property_truncated_normal}, we know the expectation of $B_i$ is 
\begin{equation*}
\E[B_i] = 0 + (1 - \Phi(a)) (\E[\widetilde{\xi}_i] - a^+)\mu = 0.
\end{equation*}
 Therefore, to prove that the mean-zero variable $B_i$ is sub-Gaussian, we only need to bound the probability of $\P(B_i > t)$ and $\P(B_i < - t)$ for any $t > 0$ with the form of $\exp(-\frac{t^2}{2\sigma^2})$ for some $\sigma>0$.

We will first prove the part for $\P(B_i > t)$ 
Note that
\begin{align*}
    \P(B_i > t) &= (1 - \Phi(a)) \P(\widetilde{\xi}_i\mu - a^+\mu > t)\\
    & = (1 - \Phi(a)) \P(\widetilde{\xi}_i\mu - a^+\mu > t)\\
    & = (1 - \Phi(a)) \P(\widetilde{\xi}_i>a^+ + t/\mu)\\
    & = (1 - \Phi(a)) \left(1 - \frac{(\Phi(a^+ + t/\mu) - \Phi(a))}{(1 - \Phi(a))}\right) \\
    & = 1 - \Phi(a^+ + t/\mu ) \\
    &\le 1 - \Phi(t/\mu ) \le \exp\left(-\frac{t^2}{2\mu^2}\right),
\end{align*}
where the fourth equality is due to (2) in Lemma \ref{lem:property_truncated_normal}, the first inequality is due to the fact that $a^+ \ge a > 0$, and the last inequality is due to the fact that $\mathcal{N}(0, \mu^2)$ is sub-Gaussian with parameter $\mu^2.$

We now prove the other side. Observe that $B_i > \mu (a^+ - a)$, so for $t > \mu (a^+ - a)$ we have $\P(B_i < - t) = 0$. Therefore, we only need to bound $\P(B_i < - t)$ for $0 < t <= \mu (a^+ - a)$. Note that in this range, 
\begin{align*}
    \P(B_i < - t) \le \P(B_i < 0) &= (1 - \Phi(a))\P(\widetilde{\xi}_i\mu - a^+\mu < 0) \\
    &=(1 - \Phi(a))\P(\widetilde{\xi}_i <  a^+) \\
    & = \Phi(a^+) - \Phi(a),
\end{align*}
where the last equality is again due to (2) of Lemma \ref{lem:property_truncated_normal}. On the other hand,  we have for any $\sigma>0$,
\begin{equation*}
\exp\left( - \frac{t^2}{2\sigma^2}\right) \ge \exp\left( - \frac{(\mu (a^+ - a))^2}{2\sigma^2}\right),  \text{ for any }0 < t \le \mu (a^+ - a).
\end{equation*}
And with the choice of $\sigma^2 = \mu^2\frac{(a^+ - a)^2}{2\log(\Phi(a^+) - \Phi(a))}$, we have 
\begin{align*}
    \exp\left( - \frac{t^2}{2\sigma^2}\right) &\ge \exp\left( - \frac{(\mu (a^+ - a))^2}{2\sigma^2} \right)=\Phi(a^+) - \Phi(a) \ge \P(B_i < - t)
\end{align*}
holds for all $0 < t \le \mu (a^+ - a)$.

Therefore, combining the two sides, we know that $B_i$ is sub-Gaussian with parameter $\max\{\mu^2$, $ \mu^2\frac{(a^+ - a)^2}{2\log(\Phi(a^+) - \Phi(a))}\}$, which concludes the proof.

\end{proof}

\begin{lemma}\label{lem:property_truncated_normal}
For a truncated normal distribution $\widetilde{\xi}_i$ with density $f(t) = \frac{\phi(t) }{1 - \Phi(a)}, \text{for }t > a$ we have
\begin{enumerate}
    \item $\E(\widetilde{\xi}_i) = \frac{\phi(a)}{1 - \Phi(a)}$.
    \item $\P(\widetilde{\xi}_i \le t) = \frac{\Phi(t) - \Phi(a)}{1 - \Phi(a)}$.
\end{enumerate}
\end{lemma}
\begin{proof}[Proof of Lemma \ref{lem:property_truncated_normal}]
This is based on several well-known truncated normal properties, and is easy to prove from the density function. Therefore we omit the proof here. 
\end{proof}

\section{Details of Edgeworth approximation error}
\label{appen:details_edgeworth}
The following discussion is largely adapted from~\cite{derumigny2021explicit} to be self-contained. For a distribution $P$, let $f_{P}$ denote its characteristic function; similarly, for a random variable $X$, we denote by $f_{X}$ its characteristic function. We recall that $f_{\mathcal{N}(0,1)}(t)=e^{-t^{2} / 2}$. Some constants are used in the definition.
\begin{itemize}
    \item Denote by $\chi_{1}$ the constant $\chi_{1}:=\sup _{x>0} x^{-3}\left|\cos (x)-1+x^{2} / 2\right| \approx 0.099162$ \citep{shevtsova2010refinement},
    \item  Denote by $\theta_{1}^{*}$ the unique root in $(0,2 \pi)$ of the equation $\theta^{2}+2 \theta \sin (\theta)+6(\cos (\theta)-1)=0$,
    \item Denote by  $t_{1}^{*}:=\theta_{1}^{*} /(2 \pi) \approx 0.635967$ \citep{shevtsova2010refinement}.
\end{itemize}

\subsection{Details of first-order Edgeworth expansion}
\label{appendx:e1}

We now provide details on the first-order Edgeworth expansion in Lemma \ref{lem:first_order_edgeworth_bound}. The main narrative is adapted from \cite{derumigny2021explicit}. 

We first define the reminder term $r_{1,m}$. To this end, we define \begin{equation*}
\Psi(t):=\frac{1}{2}\left(1-|t|+i\left[(1-|t|) \cot (\pi t)+\frac{\operatorname{sign}(t)}{\pi}\right]\right) \mathbbm{1}\{|t| \leq 1\}
\end{equation*}
where $i$ is the imaginary number. 
Note that from \cite{prawitz1975remainder} we have the following bound for function $\Psi$:\begin{equation*}
    |\Psi(t)| \leq \frac{1.0253}{2 \pi|t|} \text { and }\left|\Psi(t)-\frac{i}{2 \pi t}\right| \leq \frac{1}{2}\left(1-|t|+\frac{\pi^{2}}{18} t^{2}\right).
\end{equation*}
We further define
\begin{align*}
I_{3,1}(T):&=\frac{2}{T} \int_{0}^{\sqrt{2 \varepsilon}\left(m / K_{4, m}\right)^{1 / 4}}|\Psi(u / T)|\left|f_{S_{m}}(u)-e^{-u^{2} / 2}\left(1-\frac{i u^{3} \lambda_{3, m}}{6 \sqrt{m}}\right)\right| d u\\
I_{3,2}(T):&=\frac{2}{T} \int_{\sqrt{2 \varepsilon}\left(m / K_{4, m}\right)^{1 / 4}}^{t_{0} T}|\Psi(u / T)|\left|f_{S_{m}}(u)-e^{-u^{2} / 2}\right| d u\\
I_{3,3}(T):&=\frac{2}{T} \frac{\left|\lambda_{3, m}\right|}{6 \sqrt{m}} \int_{\sqrt{2 \varepsilon}\left(m / K_{4, m}\right)^{1 / 4}}^{t_{0} T}|\Psi(u / T)| e^{-u^{2} / 2}|u|^{3} du, 
\end{align*}
and $r_{1, m}$ is defined to be
\begin{align}
r_{1, m} &:=\frac{(14.1961+67.0415) \widetilde{K}_{3, m}^{4}}{16 \pi^{4} m^{2}}+\frac{\left|\lambda_{3, m}\right| \exp \left(-2 m^{2} / \widetilde{K}_{3, m}^{4}\right)}{3 \pi \sqrt{m}}+I_{3,2}(T)+I_{3,3}(T) \notag\\
&+\frac{1.0253}{\pi} \int_{0}^{\sqrt{2 \varepsilon}\left(n / K_{4, m}\right)^{1 / 4}} u e^{-u^{2} / 2} R_{1, m}(u, \epsilon) d u.\label{eq:r_1,n}
\end{align}
For $\epsilon\in(0,1 / 3)$ and $t \geq 0$, we further define

{\footnotesize
\begin{align*}
R_{1,m}(t, \varepsilon):&=\frac{U_{1,1, m}(t)+U_{1,2, m}(t)}{2(1-3 \varepsilon)^{2}}
+e_{1}(\varepsilon)\left(\frac{t^{8} K_{4, m}^{2}}{2 m^{2}}\left(\frac{1}{24}+\frac{P_{1, m}(\varepsilon)}{2(1-3 \varepsilon)^{2}}\right)^{2}+\frac{|t|^{7}\left|\lambda_{3, m}\right| K_{4, m}}{6 m^{3 / 2}}\left(\frac{1}{24}+\frac{P_{1, m}(\varepsilon)}{2(1-3 \varepsilon)^{2}}\right)\right),\\
P_{1, m}(\varepsilon):&=\frac{144+48 \varepsilon+4 \varepsilon^{2}+\left\{96 \sqrt{2 \varepsilon}+32 \varepsilon+16 \sqrt{2} \varepsilon^{3 / 2}\right\} \mathbbm{1}\left\{\exists i \in\{1, \ldots, m\}: \mathbb{E}\left[(X_{i}-\mu_i)^{3}\right] \neq 0\right\}}{576},\\
e_{1}(\varepsilon):&=\exp \left(\varepsilon^{2}\left(\frac{1}{6}+\frac{2 P_{1, m}(\varepsilon)}{(1-3 \varepsilon)^{2}}\right)\right),\\
U_{1,1, m}(t):&=\frac{t^{6}}{24}\left(\frac{K_{4, m}}{m}\right)^{3 / 2}+\frac{t^{8}}{24^{2}}\left(\frac{K_{4, m}}{m}\right)^{2},\\
U_{1,2, m}(t):&=\left(\frac{|t|^{5}}{6}\left(\frac{K_{4, m}}{m}\right)^{5 / 4}+\frac{t^{6}}{36}\left(\frac{K_{4, m}}{m}\right)^{3 / 2}+\frac{|t|^{7}}{72}\left(\frac{K_{4, m}}{m}\right)^{7 / 4}\right) \mathbbm{1}\left\{\exists i \in\{1, \ldots, m\}: \mathbb{E}\left[(X_{i}-\mu_i)^{3}\right] \neq 0\right\} . 
\end{align*}
}

Observe the bound from Lemma \ref{lem:first_order_edgeworth_bound} is a bound of leading order $O(1/\sqrt{m})$, which is due to the fact that the variables in the sequence may not be identical since we may encounter non-identical compositions, and we do not require any continuous property of the densities (and their existence as well). When we have i.i.d.~distribution of absolute continuous density, we can guarantee to have an $O(1/m)$ explicit bound of the difference as 
\begin{equation*}
\Delta_{m, 1} \leq \frac{0.195 K_{4, m}+0.038 \lambda_{3, m}^{2}}{m}+\frac{1.0253}{\pi} \int_{a_{m}}^{b_{m}} \frac{\left|f_{S_{m}}(t)\right|}{t} d t+r_{2, m},
\end{equation*}
where $a_{m}:=2 t_{1}^{*} \pi \sqrt{m} / \widetilde{K}_{3, m}, b_{n}:=16 \pi^{4} m^{2} / \widetilde{K}_{3, m}^{4}$, and $r_{2, m}$ is a remainder term that depends only on $K_{3,m}, K_{4, m}$ and $\lambda_{3, m}$. Specifically, the term $r_{2, m}$ is defined by
\begin{equation*}
\begin{aligned}
r_{2, m} &:=\frac{1.2533 \widetilde{K}_{3, m}^{4}}{16 \pi^{4} m^{2}}+\frac{0.3334\left|\lambda_{3, m}\right| \widetilde{K}_{3, m}^{4}}{16 \pi^{4} m^{5 / 2}}+\frac{14.1961 \widetilde{K}_{3, m}^{16}}{16^{4} \pi^{16} m^{8}}+\frac{\left|\lambda_{3, m}\right| \exp \left(-128 \pi^{6} m^{4} / \widetilde{K}_{3, m}^{8}\right)}{3 \pi \sqrt{m}} \\
&+I_{5,2}(T)+I_{5,3}(T)+I_{5,4}(T)+J_{3}(T)+J_{5}(T) \\
&+\frac{1.0253}{\pi} \int_{0}^{\sqrt{2 \varepsilon}\left(m / K_{4, m}\right)^{1 / 4}} u e^{-u^{2} / 2} R_{1, m}(u, \varepsilon) d u .
\end{aligned}
\end{equation*}
Here, 
\begin{equation*}
\begin{aligned}
&I_{5,2}(T):=E_{1, m} \frac{\left|\lambda_{3, m}\right|}{3 T \sqrt{m}} \int_{\sqrt{2 \varepsilon}\left(m / K_{4, m}\right)^{1 / 4}}^{T^{1 / 4} / \pi}|\Psi(u / T)| u^{3} e^{-u^{2} / 2} d u, \\
&I_{5,3}(T):=E_{1, m} \frac{2}{T} \int_{\sqrt{2 \varepsilon}\left(m / K_{4, m}\right)^{1 / 4}}^{T^{1 / 4}}|\Psi(u / T)|\left|f_{S_{m}}(u)-e^{-u^{2} / 2}\right| d u, \\
&I_{5,4}(T):=E_{2, m} \frac{\left|\lambda_{3, m}\right|}{3 T \sqrt{m}} \int_{T^{1 / 4} / \pi}^{T / \pi}|\Psi(u / T)||u|^{3} e^{-u^{2} / 2} d u,
\end{aligned}
\end{equation*}
where $E_{1, m}:=1_{\left\{\sqrt{2 \varepsilon}\left.\left(m / K_{4, m}\right)^{1 / 4}<T^{1 / 4} / \pi\right\}\right.}$ and $E_{2, m}:=1_{\left\{T^{1 / 4}<T\right\}} .$ Further,  $T=16 \pi^{4}m^{2} / \widetilde{K}_{3, m}^{4}$. Note that if $T^{1 / 4}>T$ or $\sqrt{2 \varepsilon}\left(m / K_{4, m}\right)^{1 / 4}>T^{1 / 4} / \pi$, our bounds are still valid and can even be improved in the sense that the corresponding integrals can be removed. Further, we have the following bound for the terms $I_{5,2}$, $I_{5,3}$ and $I_{5,4}$:
\begin{equation*}
\begin{aligned}
I_{5,2}(T) & \leq \frac{\left|\lambda_{3, m}\right|}{3 \sqrt{m}} \int_{\sqrt{2 \varepsilon}\left(m / K_{4, m}\right)^{1 / 4}}^{T^{1 / 4} / \pi} \frac{1.0253}{2 \pi} u^{2} e^{-u^{2} / 2} d u \\
&=\frac{1.0253\left|\lambda_{3, m}\right|}{3 \pi \sqrt{2} \sqrt{m}}\left(\Gamma\left(3 / 2, \varepsilon\left(m / K_{4, m}\right)^{1 / 2}\right)-\Gamma\left(3 / 2, T^{1 / 2} / 2 \pi^{2}\right)\right),
\end{aligned}
\end{equation*}
\begin{equation*}
\begin{aligned} I_{5,3}(T) & \leq \frac{2}{T} \int_{\sqrt{2 \varepsilon}\left(m/ K_{4, m}\right)^{1 / 4}}^{T^{1 / 4} / \pi}|\Psi(u / T)| \frac{K_{3, m}}{6 \sqrt{m}}|t|^{3} \exp \left(-\frac{t^{2}}{2}+\frac{\chi_{1}|t|^{3} \widetilde{K}_{3, m}}{\sqrt{m}}+\frac{t^{2} \sqrt{K_{4, m}}}{2 \sqrt{m}}\right) d u \\ &=\frac{K_{3, m}}{3 \sqrt{m}} J_{2}\left(3, \sqrt{2 \varepsilon} /\left(m K_{4, m}\right)^{1 / 4}, T^{1 / 4} / \pi, \widetilde{K}_{3, m}, K_{4, m}, T, m\right) \end{aligned}
\end{equation*}
and 
\begin{equation*}I_{5,4}(T)=\frac{1.0253\left|\lambda_{3, m}\right|}{3 \pi \sqrt{2} \sqrt{m}}\left(\Gamma\left(3 / 2, T^{1 / 2} / 2 \pi^{2}\right)-\Gamma\left(3 / 2, T^{2} / 2 \pi^{2}\right)\right),
\end{equation*}
and all the terms converge exponentially fast to zero. Here $\Gamma(a,x)$ is the incomplete Gamma function and can be numerically evaluated. 

For the other terms, we have
\begin{equation*}
\begin{aligned}
&J_{3}(T):=\frac{2}{T} \int_{T^{1 / 4} / \pi}^{t_{1}^{*} T^{1 / 4}}|\Psi(u / T)|\left|f_{S_{m}}(u)\right| d u=\frac{2}{T^{3 / 4}} \int_{1 / \pi}^{t_{1}^{*}}\left|\Psi\left(v / T^{3 / 4}\right)\right|\left|f_{S_{m}}\left(T^{1 / 4} v\right)\right| d v, \\
&J_{4}(T):=1_{\left\{t_{1}^{*} T^{1 / 4}<T / \pi\right\}} \frac{2}{T} \int_{t_{0}^{*} T^{1 / 4}}^{T / \pi}|\Psi(u / T)|\left|f_{S_{m}}(u)\right| d u, \\
&J_{5}(T):=\frac{2}{T} \int_{T^{1 / 4} / \pi}^{T / \pi}|\Psi(u / T)| e^{-u^{2} / 2} d u .
\end{aligned}
\end{equation*}
Obviously, now all the above bounds are real integrations, and can be calculated numerically. 
\subsection{Improved $O(1/m)$ uniform bound for i.n.i.d distributed PLLRs}

In the previous section, we assume that PLLRs are identically distributed and obtain an $O(1/m)$ bound under certain regularity conditions. In fact, we can generalize such bound to the case when PLLRs are independent but not necessarily identically distributed. To this end, we resort to Corollary 3.2 in \cite{derumigny2021explicit} for the conditions on the tail behavior of characteristic functions. In particular, under a stronger assumption on the tail of the characteristic function,  we have the following Corollary as an extension for Lemma \ref{lem:first_order_edgeworth_bound}.
\begin{corollary}
Denote by $f_{\sum_{i=1}^m X_i}(t)$ as the characteristic function for the sum of random variables $X_i$. If there exists some positive constant $C_0$ such that 
$|f_{\sum_{i=1}^m X_i}(t)| \leq C_0|t|^{-2}$
for all $|t|>\frac{1}{\sqrt{n}}$, then
\begin{equation}
\Delta_{m, 1, X} \leq \frac{0.195K_{4, m}+0.038\lambda_{3, m}^2+1.0253C_0\pi^{-1}}{m}+o(m^{-1}).
\end{equation}
\end{corollary}
Note that this Corollary does not assume that $X_i$'s are identically distributed. Instead, it only assumes regularity conditions on the tail distribution of characteristic function. This gives a $O(1/m)$ bound for the general case of i.n.i.d. PLLRs.
\par We further comment here that as pointed out in \cite{derumigny2021explicit}, the assumption on characteristic function is not common in the statistical community. However, one can check that common distributions such as Gaussian distribution satisfies such regularity condition. Moreover, by \cite{ushakov1999some}, such tail condition on the characteristic functions
is satisfied whenever $P_{\sum_{i=1}^m X_i}$ has a density that is differentiable and such that its derivative is of bounded variation with total variation uniformly bounded in $m$. As a result, this provide an easy check whether a PLLR satisfies the condition and admits an $O(1/m)$ tail bound in the Edgeworth expansion.

\end{document}